\documentclass[11pt]{article}
\usepackage{fullpage}
\usepackage{times}
\usepackage{comment,amsfonts,amssymb,amsmath,amsthm,graphicx,algorithmic,algorithm}
\usepackage{thmtools, thm-restate}
\newcommand{\commentout}[1]{}
\declaretheorem{theorem}
\declaretheorem{lemma}
\declaretheorem{definition}
\declaretheorem{corollary}

\ifx\pdftexversion\undefined
\usepackage[colorlinks,linkcolor=black,filecolor=black,citecolor=black,urlco
lor=black,pdfstartview=FitH]{hyperref}
\else
\usepackage[colorlinks,linkcolor=blue,filecolor=blue,citecolor=blue,urlcolor
=blue,pdfstartview=FitH]{hyperref}
\fi

\newcommand{\alert}[1]{\textbf{\color{red}
[[[#1]]]}\marginpar{\textbf{\color{red}**}}\typeout{ALERT:
\the\inputlineno: #1}}

\newcommand{\tower}{{\rm tower}}

\newcommand{\Bin}{{\rm Bin}}

\newcommand{\N}{\mathbb{N}}

\newcommand{\R}{\mathbb{R}}

\newcommand{\poly}{{\rm poly}}

\newcommand{\E}{{\mathbb{E}}}

\newcommand{\mommit}[1]{}
\newcommand{\namedref}[2]{\hyperref[#2]{#1~\ref*{#2}}}
\newcommand{\sectionref}[1]{\namedref{Section}{#1}}

\newcommand{\theoremref}[1]{\namedref{Theorem}{#1}}

\newcommand{\remarkref}[1]{\namedref{Remark}{#1}}

\newcommand{\algref}[1]{\namedref{Algorithm}{#1}}

\newcommand{\lemmaref}[1]{\namedref{Lemma}{#1}}

\newcommand{\corollaryref}[1]{\namedref{Corollary}{#1}}

\usepackage{pdfsync}
\usepackage{authblk}

\begin{document}
\title{Prioritized Metric Structures and Embedding}

\author[1]{Michael Elkin}
\author[1]{Arnold Filtser}
\author[1]{Ofer Neiman}

\affil[1]{Department of Computer Science, Ben-Gurion University of the Negev,
Beer-Sheva, Israel. Email: \texttt{\{elkinm,arnoldf,neimano\}@cs.bgu.ac.il}}

\date{}
\maketitle

\begin{abstract}
Metric data structures (distance oracles, distance labeling schemes, routing schemes) and low-distortion embeddings provide a powerful
 algorithmic methodology, which has been successfully applied for approximation algorithms \cite{llr}, online algorithms \cite{BBMN11}, distributed algorithms \cite{KKMPT12} and  for computing sparsifiers \cite{ST04}.
 However, this methodology appears to have a limitation: the worst-case performance inherently depends on the cardinality of the metric, and one could not specify in advance which vertices/points should enjoy a better
 service (i.e., stretch/distortion, label size/dimension) than that given by the worst-case guarantee.

 In this paper we alleviate this limitation by
 devising a suit of {\em prioritized} metric data structures and embeddings. We show that given a priority ranking $(x_1,x_2,\ldots,x_n)$
 of the graph vertices (respectively, metric points) one can devise a metric data structure (respectively, embedding) in which the stretch
(resp., distortion) incurred by any pair containing a vertex $x_j$ will depend on the rank $j$ of the vertex. We also show that other important
parameters, such as the label size and (in some sense) the dimension, may depend only on $j$. In some of our
metric data structures  (resp., embeddings) we achieve both prioritized stretch (resp., distortion) and label size (resp.,
dimension) {\em simultaneously}.
The worst-case performance of our metric data structures and embeddings is typically asymptotically no worse than of their non-prioritized counterparts.
\end{abstract}

\thispagestyle{empty}
\newpage
\setcounter{page}{1}
\section{Introduction}
The celebrated distance oracle of Thorup and Zwick \cite{TZ01} enables one to preprocess an undirected weighted $n$-vertex graph
$G= (V,E)$ so that to produce a data structure (aka {\em distance oracle}) of size $O(t \cdot n^{1+1/t})$ (for a parameter $t = 1,2,\ldots$)
that supports distance queries between pairs $u,v \in V$ in time $O(t)$ per query. (The query time was recently improved to $O(1)$ by
\cite{C14,WN13}.) The distance estimates provided by the oracle are within a factor of $2t-1$ from the actual distance $d_G(u,v)$
between $u$ and $v$ in $G$. The approximation factor ($2t-1$ in this case) is called the {\em stretch}. Distance oracles can serve as an example of a {\em metric data structure}; other very well-studied
examples include {\em distance labeling} \cite{Peleg99,GPPR01} and {\em routing} \cite{TZ01a,AP92}.
Thorup-Zwick's oracle can also be converted into a distance-labeling scheme: each vertex is assigned a label of size
$O(n^{1/t}\cdot\log ^{1-1/t}n )$ so that given labels of $u$ and $v$ the query algorithm can provide a $(2t-1)$-approximation of $d_G(u,v)$.
Moreover, the oracle also gives rise to a routing scheme \cite{TZ01a} that exhibits a similar tradeoff.

A different but closely related thread of research concerns {\em low-distortion embeddings}. A celebrated theorem of Bourgain
\cite{bourgaintrees} asserts that any $n$-point metric $(X,d)$ can be embedded into an $O(\log n)$-dimensional Euclidean space with
{\em distortion} $O(\log n)$. (Roughly speaking, distortion and stretch are the same thing. See \sectionref{sec:preliminaries} for formal definitions.)
 Fakcharoenphol et al. \cite{FRT} (following Bartal \cite{bartal1,bartal2}) showed that any mertic $(X,d)$ embeds into a distribution over trees (in fact,
 ultrametrics) with expected distortion $O(\log n)$.

These (and many other) important results are not only appealing from a mathematical perspective, but they also were found extremely useful
for numerous applications in Theoretical Computer Science and beyond \cite{llr,BBMN11,KKMPT12,ST04}. A natural disadvantage is the
dependence of all the relevant parameters on $n$, the cardinality of the input graph/metric. However, all these results are either
completely tight, or very close to being completely tight. In order to address this issue, metric data structures and embeddings in which
some pairs of vertices/points enjoy better stretch/distortion or with improved label size/dimension were developed. Specifically, \cite{KSW04,ABCD05,ABN06,CDG06} studied
embeddings and distance oracles in which the distortion/stretch of at least $1-\epsilon$ fraction of the pairs is improved as a function
of $\epsilon$, either for a fixed $\epsilon$ or for all $\epsilon\in[0,1]$ simultaneously
(e.g.  for a fixed $\epsilon$, embeddings into Euclidean space of dimension $O(\log 1/\epsilon)$ with distortion $O(\log(1/\epsilon))$,
or a distance oracle with stretch $2\lceil t \cdot {{\log (2/\epsilon)} \over {\log n}} \rceil + 1$ for $1-\epsilon$ fraction of the pairs).
Also, \cite{ABN07,SS09,AC14} devised embeddings and distance oracles that provide distortion/stretch $O(\log k)$ for all pairs $(x,y)$ of points such that $y$ is among the $k$ closest points to $x$, and
distance labeling schemes that support queries only between $k$-nearest neighbors, in which the label size depends only
on $k$ rather than $n$.

An inherent shortcoming of these results is, however, that the pairs that enjoy better than worst-case distortion cannot be specified in advance.
In this paper we alleviate this shortcoming and devise a suit of prioritized metric data structures
and low-distortion embeddings. Specifically, we show that one can order the graph vertices $V = (x_1,\ldots,x_n)$ {\em arbitrarily in
advance}, and devise metric data structures  (i.e., oracles/labelings/routing schemes) that, for a parameter $t= 1,2,\ldots$, provide
stretch $2\lceil t \cdot {{\log j} \over {\log n}} \rceil - 1$ (instead of $2t-1$) for {\em all pairs involving $x_j$}, while using the
same space as corresponding non-prioritized  data structures! In some cases the label size can be simultaneously improved for the high priority
points, as described in the sequel.

The same phenomenon occurs for low-distortion embeddings. We devise an embedding of general metrics into an $O(\log n)$-dimensional
Euclidean space that provides {\em prioritized distortion} $O(\log j \cdot (\log\log j)^{1/2 + \epsilon})$, for any constant $\epsilon>0$
(i.e., the distortion for all pairs containing  $x_j$ is $O(\log j \cdot (\log\log j)^{1/2 + \epsilon})$).
 Similarly, our embedding into a distribution of trees provides prioritized expected distortion $O(\log j)$.

We introduce a novel notion of {\em improved dimension} for high priority points. In general we cannot expect that the dimension of
an Euclidean embedding with low distortion (even prioritized) will be small (as Euclidean embedding into dimension $D$ has worst-case distortion of
 $\Omega(n^{1/D}\cdot\log n)$ for some metrics \cite{ABN06}). What we can offer is an embedding in which the high ranked points have only
 a few "active" coordinates. That is, only the first $O(\poly(\log j))$ coordinates in the image of $x_j$ will be nonzero, while the distortion
 is also bounded by $O(\poly(\log j))$. This could be
 useful in a setting where the high ranked points participate in numerous computations, then since representing these points requires very
 few coordinates, we can store many of them in the cache or other high speed memory.
We remark that our framework is {\em the first} which allows simultaneously improved
distortion and dimension (or improved stretch and label size) for the high priority points, while providing some guarantee for all pairs.

We have a construction of prioritized distance oracles that exhibits a qualitatively different behavior than of our aforementioned oracles. Specifically, we devise a
distance oracle with space $O(n \log \log n)$ (respectively, $O(n \log^* n)$) and prioritized stretch $O({{\log n} \over {\log (n/j)}})$
(respectively, $2^{O({{\log n} \over {\log (n/j)}})}$). Observe that as long as $j < n^{1-\epsilon}$ for any fixed $\epsilon > 0$, the
prioritized stretch of both these oracles is $O(1)$. The query time is $O(1)$. These oracles are, however,  not path-reporting (a path reporting oracle can return an actual approximate
shortest path in the graph, in time proportional to its length). We also
devise a path-reporting prioritized oracle, which was mentioned above: it has space $O(t \cdot n^{1+1/t})$, stretch
$2 \lceil t \cdot {{\log j} \over {\log n}} \rceil - 1$, and the query time\footnote{We believe the query time can be improved to
$O(1)$: \cite{C14} combines the oracles of \cite{TZ01} and of \cite{MN06} to obtain query time $O(1)$. In the full version of our paper
we show that the oracle of \cite{MN06} can be altered to give prioritized stretch, similar to that of \cite{TZ01} we show here. Using the
techniques of \cite{C14} should thus yield prioritized stretch with $O(1)$ query time.}   is $O(t \cdot {{\log j} \over {\log n}})$. In the
full version of this paper we also devise  a path-reporting prioritized distance oracle (extending \cite{EP15}) with space $O(n \log\log n)$,
stretch $O( ({{\log n} \over {\log (n/j)}})^{\log_{4/3} 7})$, and query time $O(\log ( {{\log n} \over {\log (n/j)}}))$. (Observe that this
stretch and query time are  $O(1)$ for all $j \le n^{1-\epsilon}$.)

This second oracle can be distributed as a labeling scheme, in which not only the stretch $2 \lceil t \cdot {{\log j} \over {\log n}} \rceil - 1$
is prioritized, but also the label size is smaller for high priority points:
it is $O(n^{1/t} \cdot \log j)$ rather than the non-prioritized $O(n^{1/t} \cdot \log n)$. In our routing scheme, if $j$ is the priority rank
of the destination $x_j$, it has prioritized stretch $4 \lceil t \cdot {{\log j} \over {\log n}} \rceil - 3$
(instead of $4t-5$),
the routing tables have size $O(n^{1/t} \cdot \log j)$ (instead of $O(n^{1/t} \cdot \log n)$), and labels have size $O(\log j \cdot
\lceil t {{\log j} \over {\log n}} \rceil)$
(instead of $O(t \cdot \log n)$).

We also consider the dual setting in which the stretch is fixed, and label size $\lambda(j)$ of $x_j$ is smaller when $j\ll n$. The function $\lambda(j)$ will be called
{\em prioritized label size}. Specifically, with prioritized label size $O(j^{1/t} \cdot \log j)$ we can
have stretch $2t-1$. For certain points on the tradeoff curve we can even have both stretch and label size prioritized simultaneously! In
particular, a variant of our distance labeling scheme provides a prioritized stretch $2 \lceil \log j \rceil -1$ and prioritized label size $O(\log j)$.
For routing we have similar gaurantees independent of $n$.
We also devise a distance labeling scheme for graphs that exclude a fixed minor with stretch $1 + \epsilon$ and prioritized label size $O(1/\epsilon \cdot \log j)$ (extending \cite{AG06,T01}).

Another notable result in this context is our prioritized embedding into a {\em single tree}.  It is well-known that any metric can be embedded into a single dominating tree with linear distortion, and that it is tight \cite{RR98}. We show that any $n$-point metric $(X,d)$ enjoys an embedding into a single dominating tree  with prioritized
distortion $\alpha(j)$ {\em if and only if} the sum of reciprocals $\sum_{j=1}^\infty 1/\alpha(j)$ converges. In particular, prioritized distortion $\alpha(j) = j \cdot \log j \cdot (\log\log j)^{1.01}$ is admissible, while $\alpha(j) = j \cdot \log j \cdot \log\log j$ is not, i.e., both our upper and lower bounds are tight. This lower bounds stands out as it shows that it is not  always possible to replace non-prioritized distortion of $\alpha(n)$ by a prioritized distortion $\alpha(j)$. For single-tree embedding the non-prioritized distortion is linear, while the prioritized one is provably superlinear.

\subsection{Overview of Techniques}

We elaborate briefly on the methods used to obtain our  results.

\paragraph{Distance Oracles, Distance Labeling and Routing.}

We have two basic techniques for obtaining distance oracles with prioritized stretch. The first one is manifested in \theoremref{thm:simple-oracle}, and the idea is as follows: Partition the vertices into sets according to their priority, and for each set $K\subseteq V$, apply as a black-box a known distance oracle on $K$, while for the other vertices store the distance to their nearest neighbor in $K$. We show that the stretch of pairs in $K\times V$ is only a factor of 2 worse than the one guaranteed for $K\times K$. Furthermore, we exploit the fact that for sets $K$ of small size, we can afford very small stretch and still maintain small space. The exact choice of the black-box oracle and of the partitions enables a range of tradeoffs between space and prioritized stretch.

Our second technique for an oracle with prioritized stretch, used in \theoremref{thm:prior-stretch}, is based on a non-black-box variation of the \cite{TZ01} oracle.
In their construction for stretch $2t-1$, a (non-increasing) sequence of $t-1$ sets is generated by repeated random sampling. We show that if a vertex is chosen $i$ times, then the query algorithm can be changed to improve the stretch from $2t-1$ to $2(t-i)-1$, for {\em any pair} containing such a vertex. This observation only shows that there exists a priority ranking for which the oracle has the required prioritized stretch. In order to handle {\em any} given ranking, we alter the construction by forcing high ranked elements to be chosen numerous times, and show that this increases the space usage by at most a factor of 2.

In order to build a distance labeling scheme out of their distance oracle, \cite{TZ01} pay an additional factor of $O(\log^{1-1/t} n)$ in the label size (which essentially comes from applying concentration bounds). Attempting to circumvent this logarithmic dependence on $n$, in \theoremref{thm:label} we give a different bound on the deviation probability that depends on the priority ranking of the point. Thus the increase in the label size for the $j$-th point in the ranking is only $O(\log j)$. To obtain arbitrary fixed stretch $2t-1$ for all pairs,  in \theoremref{thm:full-prior-label} we combine this scheme with an iterative application of a {\em source restricted} distance labeling of \cite{RTZ05}.

Most results on distance labeling for bounded treewidth graphs, planar graphs, and graphs excluding a fixed minor, are based on recursively partitioning the graph into small pieces using small separators (as in \cite{LT79}). The label of a vertex essentially consists of the distances to (some of) the vertices in the separator. In order to obtain prioritized label size, such as those given in \theoremref{thm:label-separate} and \theoremref{thm:app-planar}, high ranked vertices should participate in few iterations. To this end, we define multiple phases of applying separators, where each phase tries to separate only certain subset of the vertices (starting with the highest ranked, and finishing in the lowest). This way high ranked vertices will belong to a separator after a few levels, thus their label will be short.

Tree-routing of \cite{T01} is based on categorizing tree vertices as either heavy or light, depending on the size of their subtree. Our prioritized tree-routing assigns each vertex a weight which depends on its priority, and a vertex is heavy if the sum of weights of its descendents is sufficiently large. This idea paves way to our prioritized routing scheme for general graphs as well.

\paragraph{Embeddings}

It is folklore that a metric minimum spanning tree (henceforth, MST) achieves distortion $n-1$.  For our prioritized embedding of general metrics $(X,d)$ into a single tree we consider a complete graph $G = (X,{X \choose 2})$  with weight function that depends on the priority ranking. Specifically, edges incident on high-priority points get higher weights. We then compute an MST in this (generally non-metric) graph, and show that, given a certain convergence condition on the priority ranking, this MST provides a desired prioritized single-tree embedding. Remarkably, we also show that when this condition is not met, no such an embedding is possible even for the metric induced by $C_n$. Hence this embedding is tight.

Our probabilistic embedding to trees with prioritized expected distortion in \theoremref{thm:strong-ultra} is based on the construction of \cite{FRT}. The method of \cite{FRT} involves sampling a random permutation and a random radius, then using these to create a hierarchical partitioning of the metric from which a tree is built. We make the observation that, in some sense, the expected distortion of a point depends on its position in the permutation. Rather than  choosing a permutation uniformly at random, we choose one which is strongly correlated with the given priority ranking. One must be careful to allow sufficient randomness in the permutation choice so that the analysis can still go through, while guaranteeing that high ranked points will appear in the first positions of the permutation.

The embedding of \theoremref{thm:lp-strong-forall} for arbitrary metrics $(X,d)$ into Euclidean space (or any $\ell_p$ space) with prioritized distortion uses similar ideas. We partition the points to sets according to the priorities, for every set $K\subseteq X$ apply as a black-box the embedding of \cite{bourgain}. We show that since the embedding has certain  properties, it can be extended in a Lipschitz manner to all of the metric, while having distortion guarantee for any pair in $K\times X$.

The result of \theoremref{thm: Priority distortionn and space}, which gives prioritized distortion and dimension, is more technically involved. In order to ensure that high priority points are mapped to the zero vector in the embeddings tailored for the lower priority points, we change Bourgain's embedding, which is defined as distances to randomly chosen sets. Roughly speaking, when creating the embedding for a set $K$, we add all the higher ranked points to the random sets. This means the original analysis does not work directly, and we turn to a subtle case analysis to bound the distortion; see \sectionref{sec:prior-dimension} for more details.

\subsection{Organization}

After a few
preliminary definitions, we show the single tree prioritized embedding in \sectionref{sec:tree}, and the probabilistic version in \sectionref{sec:ultrametric}. In \sectionref{sec:oracle}  we discuss our prioritized distance oracles, and in \sectionref{sec:label}  the prioritized labeling schemes. The prioritized routing is shown in \sectionref{sec:route}. Finally, in \sectionref{sec:embed} we present our prioritized embedding results into normed spaces.

\section{Preliminaries}\label{sec:preliminaries}

All the graphs $G=(V,E)$ we consider are undirected and weighted. Let $x_1,\dots,x_n\in V$ be a priority ranking of the vertices. Let $d_G$ be the shortest path metric on $G$, and let $\alpha,\beta:[n]\to\R_+$ be monotone non-decreasing functions.

A distance oracle for a graph $G$ is a succinct data structure, that can approximately report distances between vertices of $G$. The parameters of this data structure we will care about are its space, query time, and stretch factor. We always measure the space of the oracle as the number of words needed to store it (where each word is $O(\log n)$ bits). The oracle has {\em prioritized stretch} $\alpha(j)$, if for any $1\le j<i\le n$, when queried for $x_j,x_i$ the oracle reports a distance $\tilde{d}(x_j,x_i)$ such that
\[
d_G(x_j,x_i)\le\tilde{d}(x_j,x_i)\le \alpha(j)\cdot d_G(x_j,x_i)~.
\]
Some oracles can be distributed as a labeling scheme, where each vertex is given a short label, and the approximate distance between two
vertices should be computed by inspecting their labels alone. We say that the a labeling scheme has {\em prioritized label size} $\beta(j)$,
 if for every $j\in[n]$, the label of $x_j$ consists of at most $\beta(j)$ words. See \sectionref{sec:route} for the precise settings of routing that we consider.

Let $(X,d_X)$ be a finite metric space, and let $x_1,\dots,x_n$ be a priority ranking of the points in $X$. Given a target metric $(Y,d_Y)$, and a non-contractive map $f:X\to Y$,\footnote{The map $f$ is non-contractive if for any $u,v\in X$, $d_X(u,v)\le d_Y(f(u),f(v))$.} we say that $f$ has {\em priority distortion} $\alpha(j)$ if for all $1\le j<i\le n$,
\[
d_Y(f(x_j),f(x_i))\le \alpha(j)\cdot d_X(x_j,x_i)~.
\]
Similarly, if $f:X\to Y$ is non-expansive, then it has priority distortion $\alpha(j)$ if for all $1\le j<i\le n$, $d_Y(f(x_j),f(x_i))\ge d_X(x_j,x_i)/\alpha(j)$.
For probabilistic embedding, we require that each map in the support of the distribution is non-contractive, and the prioritized bound on the distortion holds in expectation.

In the special case that the target metric is a normed space $\ell_p$, we say that the embedding has {\em prioritized dimension} $\beta(j)$, if for every $j\in[n]$, only the first $\beta(j)$ coordinates in $f(x_j)$ may be nonzero.

\section{Single Tree Embedding with Prioritized Distortion}\label{sec:tree}

In this section we show tight bounds on the priority distortion for an embedding into a single tree. The bounds are somewhat non-standard, as they are not attained for a single specific function, but rather for the following family of functions.
Define $\Phi$ to be the family of functions $\alpha:\N\to\R_+$ that satisfy the following properties:
\begin{itemize}
\item $\alpha$ is non-decreasing.
\item $\sum_{i=1}^\infty 1/\alpha(i)\le 1$.
\end{itemize}

\subsection{Upper Bound}

\begin{figure}
\begin{centering}
\includegraphics[scale=0.8]{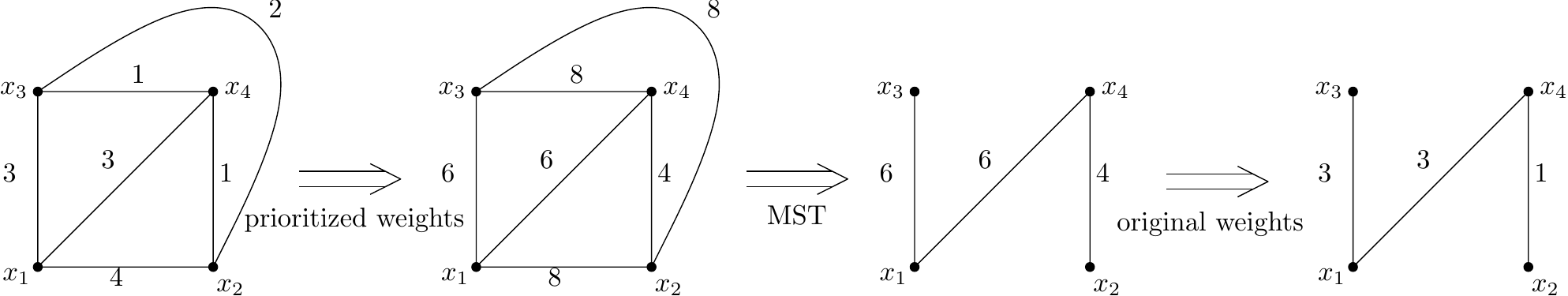}
\par\end{centering}

\caption{\label{fig:LowerBoundSingleTree}An illustration for the algorithm presented during the proof of \theoremref{thm:tree-up}.
We are given a metric space over $X=\{x_1,x_2,x_3,x_4\}$, with the function $\alpha(1)=2,\alpha(2)=4,\alpha(3)=8,\alpha(4)=16$.
In the first step we assign new weights over the edges, then find an MST in the new graph, and finally, restore the original weights.
For example the original distance between $v_2,v_3$ was $2$, while in the returned tree the distance is $7$.
Hence the pair $v_2,v_3$ suffers distortion $3.5<4$. }
\end{figure}

\begin{theorem}\label{thm:tree-up}
For any finite metric space $(X,d)$ and any $\alpha\in \Phi$, there is a (non-contractive) embedding of $X$ into a single tree with priority distortion $2\alpha(j)$.
\end{theorem}
\begin{proof}
Let $x_1,\dots,x_n$ be the priority ranking of $X$, and let $G=(X,E)$ be the complete graph on $X$. For $e=\{u,v\}\in E$, let $\ell(e)=d(u,v)$. We also define the following (prioritized) weights $w:E\to\R$, for any $1\le j<i\le n$ the edge $e=\{x_j,x_i\}$ will be given the weight $w(e)=\alpha(j)\cdot \ell(e)$. Observe that the $w$ weights on $G$ do not necessarily satisfy the triangle inequality. Let $T$ be the minimum spanning tree of $(X,E,w)$ (this tree is formed by iteratively removing the heaviest edge from a cycle). Finally, return the tree $T$ with the edges weighted by $\ell$. We claim that this tree has priority distortion $\alpha(j)$.

Consider some $x_j,x_i\in X$, if the edge $e=\{x_j,x_i\}\in E(T)$ then clearly this pair has distortion $1$. Otherwise, let $P$ be the unique path between $x_j$ and $x_i$ in $T$. Since $e$ is not in $T$, it is the heaviest edge on the cycle $P\cup \{e\}$, and for any edge $e'\in P$ we have that $w(e')\le w(e)=\alpha(j)\cdot d(x_j,x_i)$. Consider some $x_k\in X$, and note that there can be at most 2 edges touching $x_k$ in $P$. If $e'\in P$ is such an edge, and its weight by $w$ was changed by a factor of $\alpha(k)$, then $\alpha(k)\cdot\ell(e')\le \alpha(j)\cdot d(x_j,x_i)$. Summing this over all the possible values of $k$ we obtain that the length of $P$ is at most
\begin{equation}\label{eq:tree}
\sum_{e'\in P}\ell(e')\le 2\sum_{k=1}^n \frac{\alpha(j)}{\alpha(k)}\cdot d(x_j,x_i)\le 2\alpha(j)\cdot d(x_j,x_i)~.
\end{equation}
\end{proof}

\begin{corollary}\label{cor:single tree 1+eps}
For any fixed $0<\epsilon<1/2$, one can take the function $\alpha:\N\to\R$ defined by $\alpha(1)=1+\epsilon$, and for $j\ge 2$, $\alpha(j)=\frac{j(\log j)^{1+\epsilon}}{c}$, which lies in $\Phi$ for $c\approx \epsilon^2$, and obtain priority distortion $O\left(j(\log j)^{1+\epsilon}\right)$. Furthermore, the distortion of the pairs containing $x_1$ is only $1+3\epsilon$.
\end{corollary}
\begin{proof}
The fact that $\alpha\in \Phi$ follows by noting that $\int\frac{dx}{\alpha(x)}=\frac{-c}{\epsilon \cdot\log^\epsilon x}+C$. To see the small distortion for pairs $x_1,x_i$, observe that in the case $\{x_1,x_i\}\notin T$, the first edge of the path $P$ from $x_1$ to $x_i$ has weight at most $d(x_1,x_i)$, while none of the other edges on $P$ is touching $x_1$. Furthermore, since $1/\alpha(1)>1-\epsilon$, we have that $\sum_{k=2}^\infty 1/\alpha(k) <\epsilon$, and so so we can replace \eqref{eq:tree} by
\[
\sum_{e'\in P}\ell(e')\le d(x_1,x_i)+2\sum_{k=2}^n \frac{\alpha(1)}{\alpha(k)}\cdot d(x_1,x_i)\le (1+3\epsilon)\cdot d(x_1,x_i)~.
\]
\end{proof}

\subsection{Lower Bound}

Here we show a matching lower bound (up to a constant, which is only 2 for trees without Steiner nodes
\footnote{We say that the target tree has Steiner nodes if it contains more vertices than the original graph.}
on the possible functions admitting an embedding into a tree with priority distortion. We first show that a (non-decreasing) function which is not in $\Phi$ cannot bound the priority distortion in a spanning tree embedding. Then using an argument similar to that of \cite{G01}, we extend this for arbitrary dominating trees,\footnote{A tree $T$ dominates a graph $G$ if $d_T\ge d_G$.} while losing a factor of 8 in the lower bound.

\begin{theorem}\label{thm:spanning-tree-low}
For any non-decreasing function $\alpha:\N\to\R$ with $\alpha\notin\Phi$, there exists an integer $n$, a graph $G=(V,E)$ with $|V|=n$ vertices, and a priority ranking of $V$, such that no spanning tree of $G$ has priority distortion less than $\alpha$.
\end{theorem}


\begin{figure}
\begin{centering}
\includegraphics[scale=0.4]{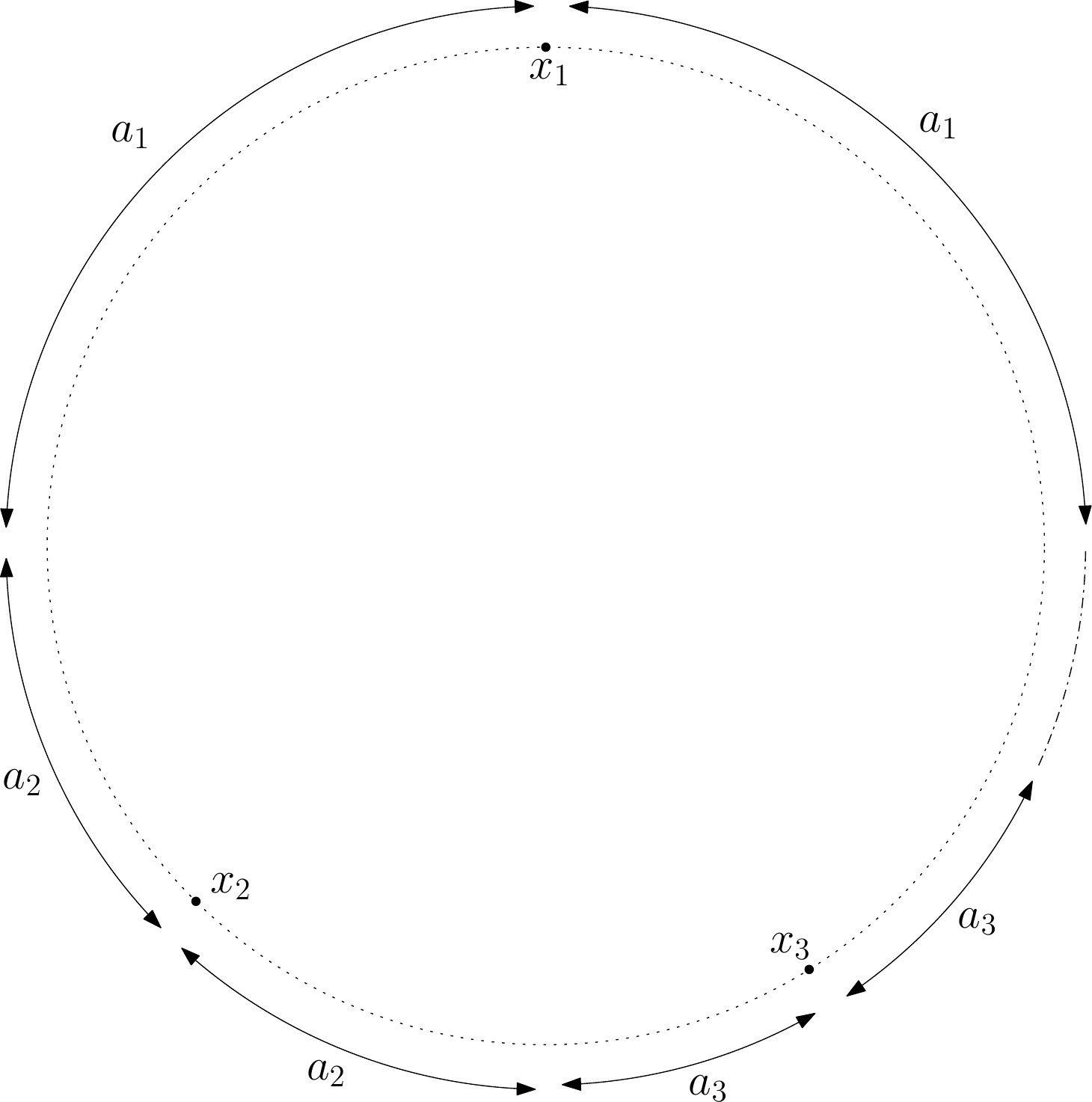}
\par\end{centering}

\caption{\label{fig:UpperBoundSingleTree}An illustration for the proof of \theoremref{thm:spanning-tree-low}. As all the pairs
containing $x_i$ cannot suffer distortion greater than $\alpha(i)$, all the edges of distance at most $a_i$ from $x_i$ cannot be
deleted from the tree.
As $\sum a_i>n$, placing $x_1,x_2,\dots$ so that the relevant
sets of edges are disjoint and cover all the edges, there is no edge that can be deleted.}
\end{figure}

\begin{proof}
Since $\alpha\notin\Phi$, there exists an integer $n'$ such that $\sum_{i=1}^{n'}1/\alpha(i)>1$. Take some integer $n>n'$ such that $\frac{n}{\alpha(i)+1}$ is an integer for all $1\le i\le n'$ (assume w.l.o.g that the $\alpha(i)$ are rational numbers). Then let $G=C_n$, a cycle on $n$ points with unit weight on the edges. Clearly, a spanning tree of $C_n$ is obtained by removing a single edge, thus we will choose the priorities $x_1,\dots,x_n\in V$ in such a way that no edge can be spared.

Seeking contradiction, assume that there exists a spanning tree with priority distortion less than $\alpha$. Let $x_1$ be an arbitrary vertex, and note that if $u$ is a vertex within distance $a_1=n/(\alpha(1)+1)$ from $x_1$, then all the edges on the shortest path from $x_1$ to $u$ must remain in the tree. Otherwise, the distortion of the pair $\{x_1,u\}$ will be at least $\frac{n-a_1}{a_1}=\alpha(1)$. There are $\frac{2n}{\alpha(1)+1}$ such edges that must belong to the tree (since we consider vertices from both sides of $x_1$). Now take $x_2$ to be a vertex at distance $\frac{n}{\alpha(1)+1}+\frac{n}{\alpha(2)+1}$ from $x_1$. By a similar argument, the $\frac{2n}{\alpha(2)+1}$ edges closest to $x_2$ must be in the tree as well. Observe that these edges form a continuous sequence on the cycle with those edges near $x_1$. Continue in this manner to define $x_3,\dots,x_{n'}$, and conclude that there are at least
\begin{equation}\label{eq:edges}
\sum_{i=1}^{n'}\frac{2n}{\alpha(i)+1}\ge \sum_{i=1}^{n'}\frac{n}{\alpha(i)}>n
\end{equation}
edges that are not allowed to be removed, but this is a contradiction, as there are only $n$ edges in $C_n$.
\end{proof}

\begin{theorem}\label{thm:tree-low}
For any non-decreasing function $\alpha:\N\to\R$ with $\alpha\notin\Phi$, there exists an integer $n$, a metric $(X,d)$ on $n$ points and a priority ranking $x_1,\dots,x_n\in X$, such that there is no embedding of $X$ into a dominating tree metric with priority distortion less than $\alpha/8$.
\end{theorem}
\begin{proof}
Take $n$, the metric $(X,d)$ induced by $C_n$, and the same priority ranking as in \theoremref{thm:spanning-tree-low}.
First consider any tree $T$ with exactly $n$ vertices, but which is not necessarily spanning. That is, $T$ is allowed to have edges that did not exist in $C_n$. Since $T$ must be dominating, we may assume that an edge in $T$ connecting vertices of distance $k$ in $C_n$ will have weight exactly $k$ (if it has larger weight, reducing it to $k$ can only improve the distortion). We extend an argument of \cite{G01} to prove that the priority distortion of $T$ is at least $\alpha$.

The argument at Section 7 of \cite{G01} says that $T$ can be replaced by a tree $T'$ satisfying $d\le d_{T'}\le d_T$, and such that any vertex in $T'$ has at most one edge to its left semicircle and one edge to its right semicircle.\footnote{If the vertices of $C_n$ are labeled $0,1,\dots,n-1$ as ordered on the cycle, the right semicircle of vertex $i$ is $\{i+1,i+2,\dots i+\lfloor n/2\rfloor\}$ (addition is modulo $n$), and the left semicircle is $V\setminus\{i,i+1,i+2,\dots i+\lfloor n/2\rfloor\}$.} A crucial observation (made in \cite{G01}) is that for any pair of vertices at distance $k$ in $C_n$, their distance in $T'$ can be either $k$ or at least $n-k$. Now we may use a similar reasoning as in the proof of \theoremref{thm:spanning-tree-low}; Assume that $x_1$ is the $i$-th vertex of $C_n$, and observe that any vertex $i+j$ for $1\le j\le a_1$, must be connected by an edge to one of the vertices $i,i+1,\dots,i+j-1$, as otherwise $d_{T'}(i,i+j)\ge n-a_1$, and the distortion of the pair $\{x_1,j\}$ will be at least $\alpha(1)$. Notice that the edges $x_2$ forces to exist are disjoint from those of $x_1$. It follows that for each $1\le i\le n'$, $x_i$ forces at least $\frac{2n}{\alpha(i)+1}$ disjoint edges to be in the tree, which is impossible due to \eqref{eq:edges}.

Finally, consider arbitrary dominating tree metrics, which may have Steiner nodes (nodes which no vertex of $C_n$ is mapped onto). By a result of \cite{G01}, such nodes may be removed while increasing the distance between any pair of points by at most 8, so we conclude that such a tree cannot have priority distortion less than $\alpha/8$.

\end{proof}

\section{Probabilistic Embedding into Ultrametrics with Prioritized Distortion}\label{sec:ultrametric}
An ultrametric $\left(U,d\right)$ is a metric space satisfying a
strong form of the triangle inequality, that is, for all $x,y,z\in U$,
$d(x,z)\le\max\left\{ d(x,y),d(y,z)\right\} $. The following definition
is known to be an equivalent one (see \cite{BLMN03}).

\begin{restatable}{definition}{Ultradef}\label{def:ultra}
An ultrametric $U$ is a metric space $\left(U,d\right)$ whose elements
are the leaves of a rooted labeled tree $T$. Each $z\in T$ is associated
with a label $\Phi\left(z\right)\ge0$ such that if $q\in T$ is a
descendant of $z$ then $\Phi\left(q\right)\le\Phi\left(z\right)$
and $\Phi\left(q\right)=0$ iff $q$ is a leaf. The distance between
leaves $z,q\in U$ is defined as $d_{T}(z,q)=\Phi\left(\mbox{lca}\left(z,q\right)\right)$
where $\mbox{lca}\left(z,q\right)$ is the least common ancestor of
$z$ and $q$ in $T$.
\end{restatable}

\begin{restatable}{theorem}{StrongUltra}\label{thm:strong-ultra}
For any metric space $(X,d)$, there exists a distribution over embeddings of $X$ into ultrametrics with expected prioritized distortion $O(\log j)$.
\end{restatable}

\begin{proof}
Let $x_1,\dots,x_n$ be the priority ranking of $X$, and let $\Delta$ be the diameter of $X$. We assume w.l.o.g that the minimal
distance in $X$ is $1$, and let $\delta$ be the minimal integer so that $\Delta\le 2^{\delta}$. We shall create a hierarchical laminar partition, where for each $i\in\{0,1,\dots,\delta\}$, the clusters of level $i$ have diameter at most $2^i$, and each of them is contained in some level $i+1$ cluster. The ultrametric is built in the natural manner, the root corresponds to the level $\delta$ cluster which is $X$, and each cluster in level $i$ corresponds to an inner node of the ultrametric with label $2^i$, whose children correspond to the level $i-1$ clusters contained in it. The leaves correspond to singletons, that is, to the elements of $X$. Clearly, the ultrametric will dominate $(X,d)$.

In order to define the partition, we choose a random permutation $\pi:X\to[n]$ which is strongly correlated with the priority ranking, and in addition we choose some number $\beta\in\left[1,2\right]$. Let $K_0=\{x_1,x_2\}$, and for any integer $1\le j\le\lceil\log\log n\rceil$ let $K_j=\{x_h~:~2^{2^{j-1}}<h\le 2^{2^j}\}$. The permutation $\pi$ is created by choosing a uniformly random permutation
on each $K_i$, and concatenating these. Note that $\pi^{-1}\left(\left\{h\in\N~:~h\in\left(2^{2^{j-1}},2^{2^{j}}\right]\right\}\right)=K_j$, and $\pi^{-1}(\{1,2\})=K_0$.

In each step $i$, we partition a cluster $S$ of level $i+1$ as follows. Each point $x\in S$ chooses the point $u\in X$ with minimal value according to $\pi$ among the points of distance at most $\beta_{i}:=\beta\cdot 2^{i-2}$
from $x$, and joins to the cluster of $u$. Note that a point may not
belong to the cluster associated with it, and some clusters may be empty (which we can discard). The description of the hierarchical partition appears in \algref{alg:mod-frt}.

\begin{algorithm}[h]
\caption{\texttt{Modified FRT}$(X,\pi)$}\label{alg:mod-frt}
\begin{algorithmic}[1]
\STATE Choose a random permutation $\pi:X\to[n]$ as above.
\STATE Choose $\beta\in\left[1,2\right]$ randomly by the distribution with the following probability density function $p\left(x\right)=\frac{1}{x\ln2}$.
\STATE Let $D_{\delta}=X$; $i\leftarrow\delta-1$.
\WHILE {$D_{i+1}$ has non-singleton clusters}
\STATE Set $\beta_{i}\leftarrow\beta\cdot 2^{i-2}$.
\FOR {$l=1,\dots,n$}
\FOR {every cluster $S$ in $D_{i+1}$}
\STATE Create a new cluster in $D_i$, consisting of all unassigned points in $S$ closer than $\beta_{i}$ to $\pi\left(l\right)$.
\ENDFOR
\ENDFOR
\STATE $i\leftarrow i-1$.
\ENDWHILE
\end{algorithmic}
\end{algorithm}

Let $T$ denote the ultrametric created by the hierarchical partition of \algref{alg:mod-frt},
and $d_{T}\left(u,v\right)$ the distance between $u$ to $v$ in
$T$. Consider the
clustering step at some level $i$, where clusters in $D_{i+1}$ are picked for partitioning. In each iteration $l$, all unassigned
points $z$ such that $d\left(z,\pi(l)\right)\le\beta_{i}$,
assign themselves to the cluster of $\pi(l)$.
Fix an arbitrary pair $\left\{ v,u\right\} $. We say that center $w$ {\em settles} the pair
$\left\{ v,u\right\} $ at level $i$, if it is the first center so that at least one of $u$ and $v$ gets assigned to its cluster. Note that exactly
one center $w$ settles any pair $\left\{ v,u\right\}$ at any particular
level. Further, we say that a center $w$ {\em cuts} the pair $\left\{ v,u\right\} $
at level $i$, if it settles them at this level,
and exactly one of $u$ and $v$ is assigned to the cluster of $w$ at level $i$.
Whenever $w$ cuts a pair $\left\{ v,u\right\} $ at level $i$, $d_{T}\left(v,u\right)$
is set to be $2^{i+1}\le 8\beta_i$. We blame this length to the
point $w$ and define $d_T^w\left(v,u\right)$ to be $\sum_{i}\mathbf{1}\left(w\mbox{ cuts \ensuremath{\left\{ v,u\right\} } at level \ensuremath{i}}\right)\cdot 8\beta_i$
(where $\mathbf{1}\left(\cdot\right)$ denotes an indicator function). We also define $d_{T}^{K_{j}}\left(v,u\right)=\sum_{w\in K_{j}}d_{T}^{w}\left(v,u\right)$.  Clearly,
$d_{T}\left(v,u\right)\le\sum_{j}d_{T}^{K_{j}}\left(v,u\right)$.

Fix some $0\le j\le\left\lceil \log\log n\right\rceil$, our next goal is to bound the expected value of $d_{T}^{K_{j}}(v,u)$ by $O\left(\log\left(\left|K_{j}\right|\right)\right)$. We arrange the points of $K_j$ in non-decreasing order of their distance
from the pair $\{v,u\}$ (breaking ties arbitrarily).
Consider the $s$th point $w_{s}$
in this sequence. W.l.o.g assume that $d\left(w_{s},v\right)\le d\left(w_{s},u\right)$.
For a center $w_{s}$ to cut $\left\{ v,u\right\} $, it must be the
case that:
\begin{enumerate}
\item $d\left(w_{s},v\right)\le\beta_{i}<d\left(w_{s},u\right)$
for some $i$.
\item $w_{s}$ settles $\{v,u\}$ at level $i$.
\end{enumerate}
Note that for each $x\in[d\left(w_{s},v\right),d\left(w_{s},u\right))$, the probability that $\beta_i\in[x,x+dx)$ is at most $\frac{dx}{x\cdot\ln 2}$. Conditioning on $\beta_i$ taking such a value $x$, any one of $w_1,\dots ,w_s$ can settle $\{v,u\}$. The probability that $w_s$ is first in the permutation $\pi$ among $w_1,\dots w_s$ is $\frac{1}{s}$. (In fact, there may be points from $\bigcup_{0\le r<j}K_{r}$ that settle $\{v,u\}$ before $w_s$. It is safe to ignore that, as it can only decrease the probability that $w_s$ cuts $\{v,u\}$.) Thus we obtain,
\begin{equation}\label{eq:w_s}
\E[d_T^{w_s}(v,u)]\le \int_{d(w_{s},v)}^{d(w_{s},u)}8x\cdot\frac{dx}{x\ln 2}\cdot\frac{1}{s}
=\frac{8}{s\cdot\ln 2}(d(w_{s},u)-d(w_{s},v))
\le\frac{16}{s}\cdot d(v,u)~.
\end{equation}
Hence we conclude,
\begin{equation}\label{eq:K_j}
\E[d_{T}^{K_{j}}(v,u)]\le\sum_{w_{s}\in K_{j}}\E[d_T^{w_s}(v,u)]\stackrel{\eqref{eq:w_s}}{\le}16d(v,u)\sum_{s=1}^{|K_{j}|}\frac{1}{s}=\log |K_{j}|\cdot O(d(v,u))~.
\end{equation}

Assume $v=x_h$ is the $h$-th vertex in the priority ranking for some $h>2$. Let $a$ be the integer such that $v\in K_a$, and recall that $2^{2^{a-1}}<h\le 2^{2^{a}}$, i.e., $2^a\le 2\log h$. The crucial observation is that if $y\in K_b$ such that $b>a$, then $y$ cannot settle $\{v,u\}$. The reason is that $v$ always appears before $y$ in $\pi$, so $v$ will surely be assigned to a cluster when it is the turn of $y$ to create a cluster. This leads to the conclusion that for all $b>a$, $\E[d_{T}^{K_{b}}(v,u)]=0$. We conclude:
\begin{eqnarray*}
\E[d_{T}(v,u)]
&\le&\sum_{j=0}^{a}\E[d_{T}^{k_{j}}(v,u)]\\
&\stackrel{\eqref{eq:K_j}}{\le}&O(d(v,u))\sum_{j=0}^{a}\log\left|K_{j}\right|\\
&=&O(d(v,u))\sum_{j=0}^{a}\log\left(2^{2^{j}}\right)\\
&=&O(d(v,u))\sum_{j=0}^{a}2^{j}\\
&=&O(d(v,u))\cdot 2^{a}\\
&=&O(d(v,u))\cdot\log h~.
\end{eqnarray*}
When $h\in\{1,2\}$ we can take $a=0$, and thus obtain a bound of $O(d(v,u))$.
\end{proof}

\section{Distance Oracles with Prioritized Stretch}\label{sec:oracle}

In this section we consider distance oracles where the stretch scales with the priority of the vertices. See \sectionref{sec:preliminaries} for the basic definitions.
A classical result of \cite{TZ01} (with improved query time due to \cite{C14}), asserts that for any parameter $t\ge 1$ and any graph on $n$ vertices, there exists a $(2t-1)$-stretch distance oracle of space $O(t\cdot n^{1+1/t})$ with $O(1)$ query time. An additional important result of \cite{MN06} allows for very small space: their oracle has space $O(n^{1+1/t})$ with stretch $O(t)$, and $O(1)$ query time as well.

\subsection{Prioritized Stretch with Small Space}\label{sec:OracleBlackBox}

Our first result provides a range of distance oracles with prioritized stretch and extremely low space. They also exhibit a somewhat non-intuitive (although very good) dependence of the stretch on the priority of the vertices. The drawbacks of these oracles are that they cannot report the approximate paths in the graph between the queried vertices, and it is not clear that they can be distributed as a labeling scheme.

For the sake of brevity, denote by $\tau(j)=\left\lfloor\frac{\log n}{\log(n/j)}\right\rfloor$ (where $n$ is always the number of vertices).
For a function $f:\N\to\N$, define its iterative application $F:\N\to\N$ as follows: $F(0)=1$, and for integer $k\ge 1$ as $F(k)=f(F(k-1))$. That is, $F(k)$ is determined by iteratively applying $f$ for $k$ times starting at 1.
\begin{theorem}\label{thm:simple-oracle}
Let $G=(V,E)$ be a weighted graph on $n$ vertices. For any positive integer $T$, let $f:\N\to\R_+$ be any monotone increasing function such that $f(1)=2$ and $F(T)\ge\log n$. Then there exists a distance oracle that requires space $O\big(\sum_{k=1}^TF(k)\cdot n\big)$ and has prioritized stretch
\[
\min\left\{4f\left(\tau(j)\right)-5,\log n\right\}~.
\]
Alternatively, one may obtain a distance oracle with space $O(T\cdot n)$ and prioritized stretch
\[
\min\left\{O\left(f\left(\tau(j)\right)\right),\log n\right\}~.
\]
Both oracles have $O(1)$ query time.
\end{theorem}

\begin{restatable}{corollary}{simoracle}\label{cor:simple-oracle}
Any weighted graph $G=(V,E)$ on $n$ vertices admits distance oracles with the following possible tradeoffs between space and prioritized
stretch.

\noindent 1) Space $O(n\log^2n)$ and prioritized stretch
$\min\{4\tau(j)-1,\log n\}$;\\ 2) Space $O(n\log n)$ and prioritized stretch $\min\{8\tau(j)-5,\log n\}$;\\ 3) Space $O(n\log\log n)$ and
prioritized stretch $\min\{O\left(\tau(j)\right),\log n\}$;\\ 4) Space $O(n\log\log\log n)$ and prioritized stretch
$\min\{O\left(\tau(j)^2\right),\log n\}$;\\ 5) Space $O(n\log^*n)$ and prioritized stretch
$\min\{O(2^{\tau(j)}),\log n\}$.

\end{restatable}

Observe that the first two oracles have stretch 3 for all points of priority less than $\sqrt{n}$, and that in all of these oracles, for any fixed $\epsilon>0$, all vertices of priority at most $n^{1-\epsilon}$ have {\em constant stretch}.

\begin{proof}[Proof of Corollary \ref{cor:simple-oracle}]
All the tradeoffs follow by simple choices for $T$ and $f$, which are described in the next bullets.
\begin{itemize}
\item For the first tradeoff let $T=\log n$ (assume w.l.o.g this is an integer), and take the function $f(k)=k+1$, so that $F(k)=k+1$ as well for all $k$. Thus the space is indeed $O(n\cdot\sum_{k=1}^T(k+1))=O(n\log^2 n)$, and the prioritized stretch is $4\tau(j)-1$ by the first assertion of \theoremref{thm:simple-oracle}.

\item For the second tradeoff, using $T=\log\log n$, it suffices to take $f(k)=2k$, so that $F(k)=2^k$. The space is now $O(n\cdot\sum_{k=1}^T2^k)=O(n\log n)$ and the prioritized stretch is as promised applying the first assertion of \theoremref{thm:simple-oracle} again.

\item In the third tradeoff we use again $T=\log\log n$, $f(k)=2k$ and $F(k)=2^k$. This time using the second assertion of \theoremref{thm:simple-oracle}, the space is $O(n\log\log n)$, and the prioritized stretch is $O\left(\tau(j)\right)$.

\item In the fourth tradeoff we use $T=1+\log\log\log n$, and let $f(1)=2$ and for $k\ge 2$, $f(k)=k^2$. It implies that $F(k)=2^{2^{k-1}}$, so by the second assertion of \theoremref{thm:simple-oracle} the space is indeed $O(n\log\log\log n)$, and the prioritized stretch follow similarly.

\item The final tradeoff holds by taking $T=\log^*n-1$, and setting $f(k)=2^k$, so that $F(k)=\tower(k)$.\footnote{$\tower(k)$ is defined as $\tower(0)=1$ and $\tower(k)=2^{\tower(k-1)}$, so that $\tower(\log^*n)=n$.} The bounds on the space and the prioritized stretch follow as before.
\end{itemize}
\end{proof}

We now turn to proving the theorem, and start with the following lemma.
\begin{lemma}\label{lem:part-oracle}
For any $t\ge 1$, and any graph $G=(V,E)$ on $n$ vertices with a subset $K\subseteq V$ of size $|K|=k$, there exists a distance oracle which can answer in $O(1)$ time queries on every pair in $K\times V$ with either:
\begin{itemize}
\item Stretch $4t-1$, using space $O(t\cdot k^{1+1/t}+n)$.
\item Stretch $O(t)$, using space $O(k^{1+1/t}+n)$.
\end{itemize}

\end{lemma}
\begin{proof}
For the first assertion, apply the distance oracle of \cite{C14} on the complete graph $G'=(K,E')$ with parameter $t$, where the weight of each edge in $E'$ is the shortest path distance in $G$ between its endpoints. This gives stretch $2t-1$ for any pair in $K\times K$ and requires space $O(t\cdot k^{1+1/t})$. For every vertex $u\in V\setminus K$, store only $d_G(u,K)$ and the name of the vertex $k_u\in K$ that manifests this distance (that is, $d_G(u,k_u)=d_G(u,K)$). We obtain a data structure of space $O(t\cdot k^{1+1/t}+n)$. To answer a distance query between $v\in K$ and $u\in V$, report $\tilde{d}(v,k_u)+d_G(k_u,u)$ where $\tilde{d}$ is the distance reported by the oracle of $G'$. It remains to bound the stretch: Observe that since $k_u$ is the closest vertex to $u$ in $K$, we have that $d_G(v,k_u)\le d_G(v,u)+d_G(k_u,u)\le 2d_G(u,v)$, and thus the reported distance is bounded as follows,
\[
\tilde{d}(v,k_u)+d_G(k_u,u)\le (2t-1)d_G(v,k_u)+d_G(u,v)\le (4t-1)d_G(u,v)~.
\]
Using the triangle inequality, the reported distance is never larger than the original,
\[
\tilde{d}(v,k_u)+d_G(k_u,u)\ge d_G(v,k_u)+d_G(k_u,u)\ge d_G(u,v)~.
\]

The second assertion follows by applying the oracle of \cite{MN06} rather than that of \cite{C14}, which yields stretch $O(t)$ on $K\times V$, and space $O(k^{1+1/t}+n)$, by a similar argument.
\end{proof}

We are finally ready to prove \theoremref{thm:simple-oracle}.
\begin{proof}[Proof of Theorem \ref{thm:simple-oracle}]
We begin with the first assertion of the theorem.
Let $x_1,\dots,x_n\in V$ the priority ranking of $V$. For each $i\in[T]$, let $S_i=\{x_j~:~ 1\le j\le n^{1-1/F(i)}\}$, and apply the first oracle of \lemmaref{lem:part-oracle} on $G$ with the set $S_i$ and parameter $t_i=F(i)-1$, let $O_i$ be the resulting oracle.\footnote{Since $F(0)=1$ and $f$ is strictly monotone, it follows that $F(i)\ge 2$ for all $i\ge 1$, so that $t_i\ge 1$.} Also invoke the oracle $O_{MN}$ of \cite{MN06} on $G$, that has stretch $\log n$ on all pairs using only $O(n)$ space (with $O(1)$ query time).

Observe that for each $i\in [T]$, the stretch $t_i$ was chosen so that $(1-1/F(i))\cdot(1+1/t_i) = 1$, so that the oracle $O_i$ has space
\[
O(t_i\cdot|S_i|^{1+1/t_i}+n) = O(F(i)\cdot n)~.
\]
The total space is thus $O\big(\sum_{i=1}^TF(i)\cdot n\big)$, as promised. It remains to prove the prioritized stretch guarantee. Fix any $v=x_j$, and let $i$ be the minimal such that $x_j\in S_i$ (observe that if $j>n/2$ there is not necessarily any such $i$). For $i=1$ the stretch guaranteed by $O_1$ is $4t_i-1=4(F(1)-1)-1=3$, as promised (recall that $f(k)\ge 2$ for all $k\ge 1$, so the required stretch is never smaller than 3). For $i>1$, by minimality of $i$ it follows that $j>n^{1-1/F(i-1)}$, that is, $F(i-1)\le\left\lfloor\frac{\log n}{\log(n/j)}\right\rfloor = \tau(j)$ (since $F(i-1)$ is an integer).  The stretch of $O_i$ for $v$ with any other point is at most
\[
4(F(i)-1)-1=4F(i)-5 = 4f(F(i-1))-5 \le 4f\left(\tau(j)\right)-5~,
\]
while the stretch of $O_{MN}$ is at most $\log n$ for all pairs, which handles the case no $i$ exists, and allows us to report the minimum of the two terms. The query time is clearly $O(1)$.

The proof of the second assertion is very similar, the only difference is using the second oracle given by \lemmaref{lem:part-oracle}. This implies oracle $O_i$ has space $O(n)$, and thus the total space is only $O(T\cdot n)$. Albeit the stretch of this oracle is worse by a constant factor.
\end{proof}


\subsection{Prioritized Distance Oracles with Bounded Prioritized Stretch}\label{sec:prior-oracles}

In this section we prove the following theorem, which prioritized the stretch of the distance oracle of \cite{TZ01}. Unlike the oracles of \theoremref{thm:simple-oracle}, this oracle can also support path queries, that is, return a path in the graph that achieves the required stretch, in time proportional to its length (plus the distance query time). Additionally, it can be distributed as a labeling scheme, which we exploit in the next section. Furthermore, this oracle matches the best known bounds for the worse-case stretch of \cite{TZ01}, which are conjectured to be optimal.

\begin{theorem}\label{thm:prior-stretch}
Let $G=(V,E)$ be a graph with $n$ vertices. Given a parameter $t\ge 1$, there exists a distance oracle of space $O(t n^{1+1/t})$ with prioritized stretch $2\lceil\frac{t\log j}{\log n}\rceil-1$ and query time $O(\lceil\frac{t\log j}{\log n}\rceil)$.
\end{theorem}


\paragraph{Overview}

Recall that in the distance oracle construction of \cite{TZ01}, a sequence of sets $V=A_0\supseteq A_1\supseteq\dots\supseteq A_t=\emptyset$ is sampled randomly, by choosing each element of $A_{i-1}$ to be in $A_i$ with probability $n^{-1/t}$. We make the crucial observation that the distance oracle provides improved stretch of $2(t-i)-1$, rather than $2t-1$, to points in $A_i$.
However, as these sets are chosen randomly, they have no correlation with our given priority list over the vertices. We therefore alter the construction, to ensure that points with high priority will surely be chosen to $A_i$ for sufficiently large $i$.

\begin{proof}[Proof of Theorem \ref{thm:prior-stretch}]
Let $x_1,\dots,x_n\in V$ be the priority ranking of $V$.
For each $i\in\{0,1,\dots,t-1\}$ let $S_i=\{x_j~:~ 1\le j \le n^{1-i/t}\}$.
Let $A_0=V$, $A_t=\emptyset$, and for each $1\le i\le t-1$ define $A'_i$ by including every element of $A_{i-1}$ with probability $n^{-1/t}/2$, and let $A_i=A'_i\cup S_i$. For each $v\in V$ and $0\le i\le t-1$, define the $i$-th pivot $p_i(v)$ as the nearest point to $v$ in $A_i$, and $B_i(v)=\{w\in A_i~:~ d(v,w) < d(v,A_{i+1})\}$.\footnote{We assume that $d(v,\emptyset)=\infty$ (this is needed as $A_t=\emptyset$).} Also the {\em bunch} of $v$ is defined as $B(v)=\bigcup_{0\le i\le t-1}B_i(v)$. The distance oracle will store in a hash table, for each $v\in V$, all the distances to points in $B(v)$, and also the $p_i(v)$ vertices.

The query algorithm for the distance between $u,v$ is essentially the same as in \cite{TZ01}, with the main difference is that we start the process at level $i$ rather than level 0, for a specified value of $i$.
\begin{algorithm}
\caption{$\texttt{Dist}(v,u,i)$}\label{alg:query}
\begin{algorithmic}[1]
\STATE $w\leftarrow v$;
\WHILE {$w\notin B(u)$}
\STATE $i\leftarrow i+1$;
\STATE $(u,v)\leftarrow (v,u)$;
\STATE $w\leftarrow p_i(v)$;
\ENDWHILE
\RETURN $d(w,u)+d(w,v)$;
\end{algorithmic}
\end{algorithm}

\paragraph{Stretch.} Let $v=x_j$ be the $j$-th point in the ordering for some $j>1$, and fix any $u\in V$. (Observe that every vertex of
$A_{t-1}$ lies in all the bunches, so when considering $x_1\in A_{t-1}$, we have that $x_1\in B(u)$ and so \algref{alg:query} will return
the exact distance.) Let $0\le i\le t-1$ be the integer satisfying that $n^{1-(i+1)/t}<j\le n^{1-i/t}$, that is, the maximal $i$ such
that $v\in S_i$. By definition we have that $v\in A_i$ as well, so we may run $\texttt{Dist}(v,u,i)$. Assuming that all operations in
the hash table cost $O(1)$, the query time is $O(t-i)$. The stretch analysis is similar to \cite{TZ01}: let $u_k$, $v_k$ and $w_k$ be the values of $u$, $v$ and $w$ at the $k$-th iteration, it suffices to show that at every iteration in which the algorithm did not stop, $d(v_k,w_k)$ increases by at most $d(u,v)$. It suffices because there are at most $t-1-i$ iterations (since $w_{t-1}\in A_{t-1}$, it lies in all bunches), so if $\ell$ is the final iteration, it must be that $d(v_\ell,w_\ell)\le (\ell-i)\cdot d(u,v)$ (initially $d(w_i,v_i)=0$), and by the triangle inequality $d(w_\ell,u_\ell)\le d(u,v)+d(v_\ell,w_\ell)\le (\ell-i+1)\cdot d(u,v)$, and as $\ell\le t-1$ we conclude that
\[
d(w,u)+d(w,v) \le (2(t-i)-1)\cdot d(u,v)~.
\]
To see the increase by at most $d(u,v)$ at every iteration, we first note that $w_i=v_i\in A_i$ (this fact enables us to start at level $i$ rather than in level 0). In the $k$-th iteration, observe that as $w_k\notin B(u_k)$ but $w_k\in A_k$, it must be that $d(u_k,p_{k+1}(u_k))\le d(u_k,w_k)$. The algorithm sets $w_{k+1}=p_{k+1}(u_k)$, $v_{k+1}=u_k$ and $u_{k+1}=v_k$, so we get that
\[
d(v_{k+1},w_{k+1}) = d(u_k,p_{k+1}(u_k))\le d(u_k,w_k)\le d(u_k,v_k) + d(v_k,w_k)= d(u,v)+d(v_k,w_k)~.
\]

Note that as $n^{1-(i+1)/t}<j\le n^{1-i/t}$, it follows that $t-i-1< \frac{t\log j}{\log n}\le t-i$, so that $t-i=\lceil\frac{t\log j}{\log n}\rceil$. The guaranteed stretch for pairs containing $x_j$ is thus bounded by $2\lceil\frac{t\log j}{\log n}\rceil-1$ (or stretch 1 for $x_1$).

\paragraph{Space.}

Fix any $u\in V$, and let us analyze the expected size of $B(u)$. Fix any $0\le i\le t-2$, and consider $B_i(u)$. Assume we have already chosen the set $A_i$, and arrange the vertices of $A_i=\{a_1,\dots a_m\}$ in order of increasing distance to $u$. Note that if $a_r$ is the first vertex in the ordering to be in $A_{i+1}$, then $|B_i(u)|=r-1$. Every vertex of $A_i$ is either in $S_{i+1}$ and thus will surely be included in $A_{i+1}$, otherwise it has probability $n^{-1/t}/2$ to be in $A'_{i+1}$ and so in $A_{i+1}$ as well. The number of vertices that we see until the first success (being in $A_{i+1}$) is stochastically dominated by a geometric distribution with parameter $p=n^{-1/t}/2$, which has expectation of $2n^{1/t}$. For the last level $t-1$, note that each vertex in $S_i\setminus S_{i+1}$ has probability exactly $(n^{-1/t}/2)^{t-1-i}=n^{-1+(i+1)/t}/2^{t-1-i}$ to be included in $A_{t-1}$, independently of all other vertices. As $|S_i\setminus S_{i+1}|\le |S_i|= n^{1-i/t}$, the expected number of vertices in $A_{t-1}$ is
\begin{equation}\label{eq:last-bunch}
\sum_{i=0}^{t-1}n^{1-i/t}\cdot n^{-1+(i+1)/t}/2^{t-1-i} <2 n^{1/t}~.
\end{equation}
This implies that $\E[|B_{t-1}(u)|]\le 2n^{1/t}$ as well, and so $\E[|B(u)|]\le 2t\cdot n^{1/t}$. The total expected size of all bunches is therefore at most $2t\cdot n^{1+1/t}$.

\end{proof}

\section{Prioritized Distance Labeling}\label{sec:label}

In this section we discuss distance labeling schemes, in which every vertex receives a short label, and it should be possible to approximately compute the distance between any two vertices from their labels alone. The novelty here is that we would like "important" vertices, those that have high priority, to have both improved stretch and also short labels.

\subsection{Distance Labeling with Prioritized Stretch and Size}\label{sec:part-label}

We begin by showing that the stretch-prioritized oracle of \theoremref{thm:prior-stretch} can be made into a labeling scheme, with the same stretch guarantees, and small label for high ranking points.
The result has some dependence on $n$ in the label size, and it seems to be interesting particularly for large values of $t$. Indeed, we shall use this result with parameter $t=\log n$ in the sequel, to obtain fully prioritized label size which will be independent of $n$, and can support any desired maximum stretch. Furthermore, this result is the basis for our routing schemes with prioritized label size and stretch.

\begin{theorem}\label{thm:label}
For any graph $G=(V,E)$ with $n$ vertices and any $t\ge 1$, there exists a distance labeling scheme with prioritized stretch $2\lceil\frac{t\log j}{\log n}\rceil-1$ and prioritized label size $O(n^{1/t}\cdot \log j )$.
\end{theorem}

\begin{proof}
Using the same notation as \sectionref{sec:oracle}, the label of vertex $v\in V$ consists of its hash table (which contains distances to all points in the bunch $B(v)$, and the identity of the pivots $p_i(v)$ for $0\le i\le t-1$). Note that \algref{alg:query} uses only this information to compute the approximate distance. The stretch guarantee is prioritized as above, and it remains to give an appropriate bound on the label sizes.

Let $x_1,\dots,x_n\in V$ be the priority ranking of $V$.
Fix a point $v=x_j$ for some $j>1$, and let $i$ be the maximal such that $v\in S_i$. Note that this implies that $t-i-1<\frac{t\log j}{\log n}$. Observe that $B_0(v)\cup\dots\cup B_{i-1}(v)=\emptyset$, so it remains to bound the size of $B_i(v),\dots,B_{t-1}(v)$. For the last set $B_{t-1}(v)=A_{t-1}$, let ${\cal E}$ be the event that $|A_{t-1}|\le 8n^{1/t}$. We already noted in \eqref{eq:last-bunch} that the expected size of $A_{t-1}$ is at most $2n^{1/t}$, thus using Markov, with probability at least $3/4$ event ${\cal E}$ holds.

For $i\le k\le t-2$, let $X_k$ be a random variable distributed geometrically with parameter $p=n^{-1/t}/2$, thus $\E[X_k]=2n^{1/t}$ for all $k$. We noted above that the distribution of $X_k$ is stochastically dominating the cardinality of $B_k(v)$, thus it suffices to bound $\sum_{k=i}^{t-2}X_k$. Observe that for any integer $s$, if $\sum_{k=i}^{t-2}X_k>s$ then it means that in a sequence of $s$ independent coin tosses with probability $p$ for heads, we have seen less than $t-1-i$ heads. That is, if $Z\sim\Bin(s,p)$ is a Binomial random variable then
\[
\Pr\left[\sum_{k=i}^{t-2}X_k>s\right] = \Pr[Z<t-1-i]\le \Pr\left[Z<\frac{t\log j}{\log n}\right]\le \Pr[Z<\log j]~.
\]
Take $s=16 n^{1/t}\cdot \log j$ (assume this is an integer), so that $\mu:=\E[Z]=8\log j$, and by a standard Chernoff bound
\[
\Pr[Z<\log j] = \Pr[Z<\mu/8]\le e^{-3\mu/8} < 1/j^3~.
\]
Let ${\cal F}=\left\{\exists~ 2\le j\le n~:~ \left|\bigcup_{k=0}^{t-2}B_k(x_j)\right|>16n^{1/t}\cdot \log j\right\}$, then by a union bound over all $2\le j\le n$ (note that the bound is non-uniform, and depends on $j$), we obtain that
\[
\Pr[{\cal F}]\le \sum_{j=2}^n\Pr\left[\left|\sum_{k=0}^{t-2}B_k(x_j)\right|>16n^{1/t}\cdot \log j\right]\le\sum_{j=2}^n 1/j^3 < 1/4~.
\]
We conclude that with probability at least $1/2$ both events ${\cal E}$ and $\bar{\cal F}$ hold, which means that the size of the bunch of each $x_j$ is bounded by $O(n^{1/t}\cdot \log j)$, as required. (Recall that $x_1\in A_{t-1}$, so its label size is $|A_{t-1}|\le 8n^{1/t}$ when event ${\cal E}$ holds.)

\end{proof}

\begin{corollary}\label{cor:part-label}
Any graph $G=(V,E)$ has a distance labeling scheme with prioritized stretch $2\lceil\log j\rceil - 1$ and prioritized label size $O(\log j)$.
\end{corollary}

\subsection{Distance Labeling with Prioritized Label Size}\label{sec:labelses}

In this section we construct a labeling scheme in which the maximum stretch is fixed for all points, and the label size is fully prioritized and independent of $n$.

\begin{restatable}{theorem}{priorlabel}\label{thm:full-prior-label}
For any graph $G=(V,E)$ and an integer $t \ge 1$, there exists a distance labeling scheme with stretch $2t-1$ and prioritized label size $O(j^{1/t}\cdot\log j)$.
\end{restatable}

\paragraph{Proof Overview.} The idea is to partition the vertices into $m:=\lceil\frac{\log n}{t}\rceil$ sets $S_1,\dots,S_m$, and to apply the result of \sectionref{sec:part-label} in conjunction with a variation of the {\em source-restricted distance oracles} of \cite{RTZ05}, using a labeling scheme rather than an oracle. In a source restricted labeling scheme on $X$ with a subset $S\subseteq X$, only distances between pairs in $S\times X$ can be queried. Replacing the source restricted oracle with a labeling scheme, demands that we use an analysis similar to \sectionref{sec:part-label} to guarantee a prioritized bound on the label sizes. We will apply this for each $i\in\{2,3,\dots,m\}$ with $X=S_i\cup\dots\cup S_m$ and the subset $S_i$. Thus an element of $S_i$ will have a label which consists of $i$ schemes, and we will guarantee that their sizes form a geometric progression, so that the total label size is sufficiently small.

As it turns out, the construction of \cite{RTZ05} is inadequate for the first $2^t$ elements $S_1$, which have very strict requirement on their label size. We will use the construction of \sectionref{sec:part-label} to handle distances involving the elements in $S_1$. Fortunately, the stretch incurred by this construction is $2\lceil\log j\rceil-1$ which is bounded by $2t-1$ for the first $2^t$ elements in the ranking. We begin by stating the source-restricted distance labeling, based on \cite{RTZ05}.

\begin{theorem}\label{thm:rtz}
For any integer $t\ge 1$, any graph $G=(V,E)$ and a subset $S\subseteq V$, there exists a source-restricted distance labeling scheme with stretch $2t-1$ and prioritized label size $O(|S|^{1/t}\cdot\log j)$.
\end{theorem}
\begin{proof}
The observation made in \cite{RTZ05} is that to obtain a source-restricted distance oracle, it suffices to sample the random sets $S=A_0\supseteq A_1\supseteq\dots\supseteq A_t=\emptyset$ only from $S$, where each element of $A_{i-1}$ is included in $A_i$ independently with probability $|S|^{-1/t}$. They show that defining the bunches as in \cite{TZ01}, the resulting stretch is $2t-1$ for all pairs in $S\times V$. We shall use a similar analysis as in \theoremref{thm:label} to argue that this can be made into a labeling scheme. The expected label size is $O(|S|^{1/t})$, and we can show that with constant probability, every point $x_j$ pays only an additional factor of $O(\log j)$. As the proof is very similar, we leave the details to the reader.
\end{proof}
\begin{proof}[Proof of Theorem \ref{thm:full-prior-label}]
Let $S_1=\{x_j~:~1\le j\le 2^t\}$, and for each $i\in\{2,3,\dots,m\}$ let $S_i=\{x_j~:~ 2^{(i-1)t} < j \le 2^{it}\}$. We have a separate construction for $i=1$ and for $i>1$. For the case $i=1$, use the labeling scheme of \corollaryref{cor:part-label} on $G=(V,E)$. For each $2\le i\le m$, apply \theoremref{thm:rtz} on $G$ and the subset $S_i$, but append the resulting labels only for vertices in $S_i\cup\dots\cup S_m$.

Fix any $u,v\in V$, and w.l.o.g assume that $v\in S_i$ has higher rank than $u$. This suggests that $u\in S_i\cup\dots\cup S_m$, thus the source restricted labeling scheme for $S_i$ guarantee stretch at most $2t-1$ for the pair $u,v$ (and $u$ indeed stored the appropriate label). Note that in the case of $v=x_j\in S_1$, the stretch can be improved to $2\lceil\log j\rceil-1$ (recall that $\log j\le t$).

We now turn to bounding the label sizes. First consider $v=x_j\in S_1$, then it must be that $j\le 2^t$. The label size of $v$ is by \corollaryref{cor:part-label} at most $O(\log j)$, and this is the final label of $v$. For $v=x_j\in S_i$ when $i\ge 2$, the label of $v$ consists of labels created for the sets $S_1,\dots,S_i$. Notice that $2^{t(i-1)}<j\le 2^{ti}$, so it holds that $2^i= (2^t\cdot 2^{t(i-1)})^{1/t}< 2j^{1/t}$. By \corollaryref{cor:part-label} the label due to $S_1$ is at most $O(\log j)$, and using \theoremref{thm:rtz} the label size of $v$ is at most
\[
O(\log j)+\sum_{k=2}^i O( |S_k|^{1/t}\cdot\log j) = O(\log j)\cdot\sum_{k=1}^i 2^k = O(2^i\cdot\log j) = O(j^{1/t}\cdot\log j)~.
\]

\end{proof}

\subsection{Prioritized Distance Labeling for Graphs with Bounded Separators}\label{sec:exact-label}

\subsubsection{Exact Labeling with Prioritized Size}
In this section we exhibit prioritized exact distance labeling scheme tailored for graphs that admit a small separator. We say that a graph $G=(V,E)$ admit an $s$-separator, if for any weight function $w:V\to\R_+$, there exists a set $U\subseteq V$ of size $|U|=s$, such that each connected component $C$ of $G\setminus U$, has $w(C)\le 2w(V)/3$.\footnote{For a set $C\subseteq V$, its weight is defined as $w(C)=\sum_{u\in C}w(u)$.} It is well known that trees admit a 1-separator, and graphs of treewidth $k$ admit a $k$-separator.

The basic idea for constructing an exact distance labeling scheme based on separators, is to create a hierarchical partition of the graph, each time by applying the separator on each connected component. Then the label of a vertex $u$ consists of all distances to all the vertices in the separators of clusters that contain $u$. To answer a query between vertices $u,v$, we return the minimum of $d(u,s)+d(v,s)$ for all separator vertices $s$ that $u,v$ have in common in their labels (this is the exact distance, because at some point a vertex on the shortest path from $u$ to $v$ must be chosen to be in a separator). Since at every iteration the number of vertices in each cluster drops by at least a constant factor, after $O(\log n)$ levels the process is complete, thus the label size is at most $O(s\log n)$.

Our improved label size for vertices of high priority, will be based on the following observation: If the weight function $w$ is an indicator for a set $S\subseteq V$ (that is, if $u\in S$ then $w(u)=1$, and if $u\in V\setminus S$ then $w(u)=0$), then after $\lceil\log|S|\rceil+1$ iterations, all vertices of $S$ must have been removed from the graph.

\begin{theorem}\label{thm:label-separate}
Let $G=(V,E)$ be a graph admitting an $s$-separator, and let $V=(x_1,\dots,x_n)$ be a priority ranking of the vertices. Then there exists an exact distance labeling scheme with prioritized label size $O(s\cdot\log j)$.
\end{theorem}
\begin{proof}
Let $S_0=\{x_1,x_2\}$, and for $1\le i\le\lceil\log\log n\rceil$ let $S_i=\{x_j~:~ 2^{2^{i-1}}<j\le 2^{2^i}\}$. The hierarchical partition will be performed in $\log\log n$ phases. The $i$-th phase consists of $2^i+1$ levels. In each level of the $i$-th phase, we generate an $s$-separator for each remaining connected component $C$, with the following weight function
\[
w(u)=\left\{\begin{array}{ccc} 1 & u\in S_i\cap C\\
0 & \textrm{otherwise} \end{array} \right.
\]
Then this separator is removed from the component.
By the observation made above, after at most $1+\log|S_i|\le 2^i+1$ levels, all remaining components have no vertices from $S_i$.
The label of a vertex $u\in V$ will be the distances to all points in the separators created for components containing $u$.

Fix some vertex $x_j$ (for $j>1$), and assume $x_j\in S_i$. Notice that $2^{i-1}<\log j$. Then the label size of $x_j$ is at most
\[
\sum_{k=0}^i s\cdot (2^k+1) = O(s\cdot 2^i)=O(s\cdot\log j)~.
\]
\end{proof}

\subsubsection{Planar Graphs and Graphs Excluding a Fixed Minor}

\underline{}While exact distance labeling for planar graphs requires polynomial label size or query time, there is a $1+\epsilon$ stretch labeling scheme for planar graphs with label size $O(\log n)$ \cite{T01,K02}, which was extended to graphs excluding a fixed minor \cite{AG06}. All these constructions are based on {\em path separators}: a constant number of shortest paths in the graph, whose removal induces pieces of bounded weight. The label of a vertex consists of distances to carefully selected vertices on these paths. We may use the same methodology as above; generate these path separators for the sets $S_i$ in order, and obtain the following.
\begin{theorem}\label{thm:app-planar}
Let $G=(V,E)$ be a graph excluding some fixed minor, and $V=(x_1,\dots,x_n)$ a priority ranking of the vertices. Then for any $\epsilon>0$ there exists a distance labeling scheme with stretch $1+\epsilon$ and prioritized label size $O((\log j)/\epsilon)$.
\end{theorem}

\section{Routing}\label{sec:route}

\subsection{Routing in Trees with Prioritized Labels}

In this section we extend a result of \cite{TZ-spaa}, and show a routing scheme on trees. The setting is that each vertex stores a routing table, and when a routing request arrives for vertex $v$, it contains $L(v)$, the label of vertex $v$. We will show the following.
\begin{restatable}{theorem}{RouteTree}\label{thm:route-tree}
For any tree $T=(V,E)$ there is a routing scheme with routing tables of size $O(1)$ and labels of prioritized size $\log j+2\log\log j+4$.
\end{restatable}

\begin{proof}
The proof follows closely the one of \cite{TZ-spaa}, with the major difference being the assignments of weights, which gives preference to the high priority vertices. Thus ensuring that when routing from the root of the tree to a vertex of rank $j$, there are $\approx\log j$ junctions that require routing information from the label of the vertex.

Let $x_1,\dots,x_n$ be the priority ranking of $V$. Let $S_0=\{x_1\}$ and for each $1\le i\le\log n$, let $S_i=\{x_j~:~2^{i-1}<j\le 2^i\}$.
Fix an arbitrary root $r$ of the tree $T$. For every $v\in S_i$ define $p(v)=\frac{1}{2^i\cdot (i+1)^2}$. Note that as $|S_i|\le 2^i$ we have that
\[
\sum_{v\in V}p(v)\le \sum_{i=0}^{\log n}\frac{2^i}{2^i\cdot(i+1)^2}\le 2~.
\]
For each $v\in V$, define the weight of $v$ as $s_v=\sum_{u\in T_v}p(u)$, where $T_v$ is the subtree rooted at $v$ (including $v$ itself). A child $v'$ of $v$ is called {\em heavy} if its weight is greater than $s_v/2$; otherwise it is called {\em light}. The root $r$ of the tree will always be considered heavy.
Observe that any vertex can have at most one heavy child.
The {\em light level} $\ell(v)$ of a vertex $v$ is defined as the number of light vertices on the path from the root to $v$, denoted by $Path(v) = (r = v_0,v_1,\ldots,v_k = v)$. The label size of $v$ will be $\ell(v)$ words.

We enumerate all vertices $T$ in DFS order, where all the light children of a vertex are visited before its heavy child is visited. (The order is otherwise arbitrary.)
We identify each vertex $v$ with its DFS number. Let $f_v$ denote the largest descendant of $v$. Also, let $h_v$ denote its heavy child, if exists. If it does not exist define $h_v  = f_v+1$.
Also, let $P(\pi(v))$ denote the port number of the edge connecting $v$ to its parent $\pi(v)$, and $P(h_v)$ denote the port number connecting $v$ to its heavy child (if it exists). The routing table stored at $v$ is $(v,f_v,h_v,P(\pi(v)),P(h_v))$. It requires $O(1)$ words.

Each time an edge from a vertex to one of its light children is taken, the weight of the corresponding subtree decreases by at least a factor of 2. Note that a vertex $v=x_j\in S_i$ has weight at least $w(v)\ge p(v)=\frac{1}{2^i\cdot (i+1)^2}$, and since the root has weight at most 2, it follows that $\ell(v)\le\log(2\cdot 2^i\cdot(i+1)^2) =i+2\log(i+1)+1$. Since $2^{i-1}<j$, we conclude that
\[
\ell(v)\le \log j +2\log(\log (j)+2)+2~.
\]

For each index $q$, $1 \le q \le \ell(v)$, denote by  $i_q$ the index of $q$-th light vertex of $Path(v)$.
Let $L(v) = (v,(port(v_{i_1-1},v_{i_1}),\ldots,port(v_{i_{\ell(v)-1}},v_{i_{\ell(v)}})))$ be the label of $v$, which consists of its name, and a sequence of at most $\ell(v)$ words containing the port numbers corresponding to the edges leading to light children on $Path(v)$.

The routing algorithm works as follows. Suppose we need to route a message with the header $L(v)$ at a vertex $w$.
The vertex $w$ checks if $w = v$. If it is the case then we are done.
Otherwise, $w$ checks if $v \in [w,f_w]$. If it is not the case, then $v$ is not in the subtree of $w$, and then $w$ sends the message to its parent.
Otherwise $w$ checks if $v \in [h_w,f_w]$. If it is the case then the message is sent to the heavy child. Otherwise $v$ is a descendent of a light child of $w$. The vertex $w$ finds itself in the sequence of $L(v)$, and determines to which light child of $w$ the message should be sent. Then it sends the message to this child.

\end{proof}

\subsection{Routing in General Graphs}

To obtain routing scheme for general graphs, we use the same method as \cite{TZ-spaa}, but replace their distance labeling with our prioritized ones from \theoremref{thm:label}. This routing scheme has the following property: after an initial calculation using the entire label of the destination vertex $v$, all routing decisions are based on a much shorter {\em header} appended to the message. In particular, we obtain the following theorem.
\begin{theorem}\label{thm:general-route}
For any graph $G=(V,E)$ with priority ranking $x_1,\dots,x_n$ of $V$, and any parameter $t\ge 1$, there exists a routing scheme, such that the label size of $x_j$ is at most $\log j\cdot \lceil \frac{t\log j}{\log n}\rceil\cdot(1+o(1))$, its header of size $\log j\cdot(1+o(1))$, and it stores a routing table of size $O(n^{1/t}\cdot\log j)$. Routing from any vertex into $x_j$ will have stretch at most $4\lceil \frac{t\log j}{\log n}\rceil-3$.
\end{theorem}

\begin{proof}[Sketch]
We use the definitions of \sectionref{sec:prior-oracles}. Consider the distance labeling scheme given in \theoremref{thm:label}.  Following \cite{TZ01}, this labeling scheme yields a tree-cover: a collection of subtrees such that vertex $v=x_j$ belongs to at most $|B(v)|$ trees. The tree $T_z$ for vertex $z$ contains $z$ as the root, and the shortest path to all the vertices in $C(z)=\{x\in V~:~ z\in B(x)\}$. To route from some vertex $u\in V$ to $v$, it suffices to find an appropriate $z\in B(u)\cap B(v)$, and route in $T_z$ by applying \theoremref{thm:route-tree}.

The {\em routing table} stored at each vertex $v\in V$ contains the hash table for its bunch $B(v)$, and the routing table needed to route in $T_z$ for each $z\in B(v)$. Recall that by \theoremref{thm:label}, $|B(v)|\le O(n^{1/t}\cdot\log j)$ (where $v=x_j$), and by \theoremref{thm:route-tree}, the routing table of each tree is of constant size. Let $i$ be the minimal such that $v=x_j\in S_i$. The {\em label} of $v$ is $\left((p_i(v),L_i(v)),\dots,(p_{t-1}(v),L_{t-1}(v))\right)$, where $L_h(v)$ is the label of $v$ that is required to route in $T_{p_h(v)}$. Note that the label is of size $(t-i)\log j\cdot(1+o(1))=\log j\cdot \lceil \frac{t\log j}{\log n}\rceil\cdot(1+o(1))$ (the equality follows from a calculation done in \sectionref{sec:prior-oracles}).

Finding the tree which guarantees the prioritized stretch as in \theoremref{thm:label} could have been achieved by using \algref{alg:query}, alas, this requires knowledge of the bunches of both vertices $u$ and $v$. It remains to see that using only the label of $v$ and the routing table at $u$, one can find a tree in the cover which has stretch at most $4\lceil\frac{t\log j}{\log n}\rceil-3$ for $u,v$ (routing in the tree does not increase the stretch).
To see this, let $i\le h\le t-1$ be the minimal such that $p_h(v)\in B(u)$. Following \cite{TZ-spaa}, we prove by induction that for each $i\le k\le h$ it holds that
\begin{itemize}
\item $d(v,p_k(v))\le 2(k-i)\cdot d(u,v)$,
\item $d(u,p_k(v))\le (2(k-i)+1)\cdot d(u,v)$.
\end{itemize}
The base case for $k=i$ holds as $v=p_i(v)$, assume for $k$, and for $k+1$ it suffices to prove the first item, as the second follows from the first by the triangle inequality. Since $k<h$ it follows that $p_k(v)\notin B(u)$, thus it must be that $d(u,p_{k+1}(u))\le d(u,p_k(v))$. Now,
\[
d(v,p_{k+1}(v))\le d(v,p_{k+1}(u))\le d(v,u)+d(u,p_{k+1}(u))\le d(v,u)+ d(u,p_k(v))\le (2(k-i)+2)\cdot d(u,v)~,
\]
where the last inequality uses the induction hypothesis. Finally, routing through the shortest path tree rooted at $p_h(v)$ will have stretch at most
\[
d(u,p_h(v))+d(p_h(v),v)\le (4(h-i)+1)\cdot d(u,v)\le (4(t-i)-3)\cdot d(u,v)=(4\bigg\lceil\frac{t\log j}{\log n}\bigg\rceil-3)\cdot d(u,v)~,
\]
using that $h\le t-1$ and that $t-i=\lceil\frac{t\log j}{\log n}\rceil$. This concludes the bound of the stretch. Note that once the vertex $p_h(v)$ is found, all other vertices on route from $u$ to $v$ only require the information $(p_h(v),L_h(v))$, which is appended to the message as a header of size $\log j\cdot(1+o(1))$.

\end{proof}

\begin{corollary}\label{cor:Routing}
Any graph $G=(V,E)$ with a priority ranking $x_1,\dots,x_n$ has a fully prioritized routing scheme, such that the label size of $x_j$ is
at most $\log ^2j\cdot(1+o(1))$, its header will be of size $\log j\cdot(1+o(1))$, and it stores a routing table of size $O(\log j)$.
Routing from any vertex into $x_j$ will have stretch at most $4\lceil\log j\rceil-3$.
\end{corollary}

\section{Prioritized Embedding into Normed Spaces}\label{sec:embed}

\subsection{Embedding with Prioritized Distortion}\label{sec:distortion-lp}

In this section we study embedding arbitrary metrics into normed spaces, where the distortion is prioritized according to the given ranking of the points in the metric. Our main result is the following

\begin{theorem}\label{thm:lp-strong-forall}
For any $p\in[1,\infty]$, $\epsilon>0$, and any finite metric space $(X,d)$ with priority ranking $X=(x_{1},\dots,x_{n})$,
there exists an embedding of $X$ into $\ell_p^{O(\log^2n)}$ with priority distortion $O(\log j\cdot(\log\log j)^{(1+\epsilon)/2})$.
\end{theorem}


\paragraph{Proof overview.}
Our improved distortion guarantee for high ranked points comes from a variation of Bourgain's embedding \cite{bourgain} of finite metric spaces into $\ell_p$ space. Bourgain's embedding is based on randomly sampling sets in various densities, and defining the coordinates as distances to these sets. Our first observation (see \lemmaref{lemma:lp-term-expand}) is sampling points only from a subset $K\subseteq X$, suffices to obtain an embedding which is non-expansive for all pairs, and has bounded contraction for pairs in $K\times X$. Furthermore, the contraction depends only on $|K|$, rather than on $|X|$.

We then use a similar strategy as in previous sections, and partition $X$ to roughly $\log\log n$ subsets $S_0,S_1,\dots ,S_{\log\log n}$, where $S_i$ is of size $\approx 2^{2^i}$. The doubly exponential size arises because for any $u,v\in S_i$, the logarithm of the ranking of $u$ and of $v$ differs by at most a factor of 2. For each $i$, we create the embedding $f_i$ that will "handle" pairs in $S_i\times X$, and concatenate all these functions $f=\bigoplus_{i=0}^{\log\log n}\alpha_i\cdot f_i$.
Without the $\alpha_i$ factor, every pair will suffer a $(\log\log n)^{1/p}$ term in the distortion due to expansion. We introduce these factors to the embedding, where $\alpha_i$ is such that $\sum_{i=0}^{\infty} \alpha_i^p\le 1$. In such a way, the function $f$ is non-expansive, but we pay a small factor of $1/\alpha_i$ in the distortion for pairs in $S_i\times X$.

\begin{lemma}\label{lemma:lp-term-expand}
Let $(X,d)$ be a metric space of size $|X|=n$, $K\subseteq X$ a subset of size $|K|=k$ and a parameter $p\in [1,\infty]$.
Then there is an non-expansive embedding of $X$ into $\ell_p^{O(\log^2 k)}$ such that the contraction of any pair in $K\times X$ is at most $O(\log k)$.
\end{lemma}
\begin{proof}

Let $m=O(\log^2 k)$, and $f:K\to\ell_p^m$ be a non-expansive embedding with contraction $\alpha=O(\log k)$ on the pairs of $K\times K$, which exists due to \cite{bourgain,llr}.
Since $f$ is a Fr\'{e}chet embedding, we claim that there is an extension $\hat{f}:X\to\ell_p^{O(\log^2 k)}$ of $f$ (that is $f(v)=\hat{f}(v)$ for all $v\in K$), such that $\hat{f}$ is also non-expansive. To see this, note that $f$ is defined as $f(x)=m^{-1/p}\bigoplus_{i=1}^{m} d\left(x,A_{i}\right)$ for some sets $A_i\subseteq K$. One can then simply define $\hat{f}(x)=m^{-1/p}\bigoplus_{i=1}^{m} d\left(x,A_{i}\right)$, which is indeed non-expansive.

Let $h:X\to\R$ be defined by $h(x)=d(x,K)$. The embedding $F:X\to\ell_p^m$ is defined by the concatenation of these maps $F=\hat{f}\oplus h$.
Since both of the maps $\hat{f},h$ are non-expansive, it follows that for any $x,y\in X$,
\[
\|F(x)-F(y)\|_{p}^{p}\le\|\hat{f}(x)-\hat{f}(y)\|_{p}^{p}+|h(x)-h(y)|^{p}\le2\cdot d(x,y)^{p}~,
\]
hence $F$ has expansion at most $2^{1/p}$ for all pairs. Let $t\in K$ and $x\in X$, and let $k_x\in K$ be such that $d(x,K)=d(x,k_x)$ (it could be that $k_x=x$).
If it is the case that $d(x,t)\le 3\alpha\cdot d(x,k_x)$ then by the single coordinate of $h$ we get sufficient contribution for this pair:
\[
\|F(t)-F(x)\|_p\ge |h(t)-h(x)|=h(x)=d(x,k_x)\ge \frac{d(x,t)}{3\alpha}~.
\]
The other case is that $d(x,t)> 3\alpha\cdot d(x,k_x)$, here we will get the contribution from $\hat{f}$. First observe that by the triangle inequality,
\begin{equation}\label{eq:gkp}
d(t,k_x)\ge d(t,x)-d(x,k_x)\ge d(t,x)(1-1/(3\alpha))\ge 2d(t,x)/3~.
\end{equation}
By another application of the triangle inequality, using that $\hat{f}$ is non-expansive, and that $f$ has contraction $\alpha$ on $K$,
we get the required bound on the contraction:
\begin{eqnarray*}
\|F(t)-F(x)\|_p&\ge&\|\hat{f}(t)-\hat{f}(x)\|_p\\
&\ge&\|\hat{f}(t)-\hat{f}(k_x)\|_p-\|\hat{f}(k_x)-\hat{f}(x)\|_p\\
&\ge&\|f(t)-f(k_x)\|_p-d(x,k_x)\\
&\ge&\frac{d(t,k_x)}{\alpha}-\frac{d(t,x)}{3\alpha}\\
&\stackrel{\eqref{eq:gkp}}{\ge}&\frac{2d(t,x)}{3\alpha}-\frac{d(t,x)}{3\alpha}\\
&=&\frac{d(t,x)}{3\alpha}~.
\end{eqnarray*}
In particular, the function $2^{-\frac{1}{p}}\cdot F$ is non-expansive for all pairs, and has contraction at most $2^{\frac{1}{p}}\cdot3\cdot\alpha=O(\log k)$
for pairs in $K\times X$.
\end{proof}

We are now ready to prove \theoremref{thm:lp-strong-forall}.
\begin{proof}[Proof of Theorem \ref{thm:lp-strong-forall}]
Let $S_0=\{x_1,x_2\}$, and for $1\le i\le\left\lceil \log\log n\right\rceil$ let $S_{i}=\left\{ x_{j}~:~ 2^{2^{i-1}}< j\le 2^{2^{i}}\right\}$.
For every $i$, let $f_i:X\to\ell_p$ be the embedding of \lemmaref{lemma:lp-term-expand} with $K=S_i$, and let $\alpha_i=c\cdot (i+1)^{-(1+\epsilon)/p}$ for sufficiently small constant $c$, so that $\sum_{i=0}^{\infty}\alpha_i^p\le 1$. Finally, define the embedding $f:X\to\ell_p$ by
\[
f=\bigoplus_{i=0}^{\left\lceil \log\log n\right\rceil }\alpha_{i}\cdot f_{i}~,
\]
To see that $f$ is indeed non-expansive, we recall that each $f_i$ is non-expansive, we obtain that for any $u,v\in X$
\[
\|f(u)-f(v)\|_p^p\le\sum_{i=0}^{\left\lceil \log\log n\right\rceil }\alpha_i^p\cdot\|f_i(u)-f_i(v)\|_p^p\le d(u,v)^p\sum_{i=0}^{\infty}\alpha_i^p\le d(u,v)^p~.
\]
For the contraction, let $v=x_j$ for some $j>1$, and take any $u\in X$. Let $i$ be the index such that $v\in S_i$, and note that $2^{i-1}<\log j$. By \lemmaref{lemma:lp-term-expand}, the embedding $f_i$ has contraction at most $O(\log |S_i|)=O(2^i)=O(\log j)$ for the pair $u,v$. Observe that $\alpha_i^p=c^p\cdot(i+1)^{-(1+\epsilon)}=\Omega\left((2+\log\log j)^{-(1+\epsilon)}\right)$, thus
\[
\|f(u)-f(v)\|_p^p\ge \alpha_i^p\cdot \|f(u)-f(v)\|_p^p\ge\Omega\left(\frac{d(u,v)^p}{(\log j)^p\cdot(2+\log\log j)^{-(1+\epsilon)}}\right)~.
\]

It is not hard to verify that $x_1$ has constant contraction with any $u$, so the prioritized distortion is $O\left(\log j\cdot(\log\log j)^{-(1+\epsilon)/p}\right)$.  Finally, since the dimension of $f_i$ is $O(\log^2|S_i|)=O(2^{2i})$, the embedding $f$ maps $X$ into $\sum_{i=0}^{\lceil\log\log n\rceil}O(2^{2i}) = O(\log^2n)$ dimensions. For $1\le p\le 2$, one may embed first into $\ell_2$, use \cite{jl} to reduce the dimension to $O(\log n)$, and then apply an embedding to $\ell_p^{O(\log n)}$, while paying a constant factor in the distortion \cite{flm}. The prioritized distortion will thus be at most $O(\log j\cdot (\log\log j)^{(1+\epsilon)/2})$.
\end{proof}

\subsection{Embedding with Prioritized Dimension}\label{sec:prior-dimension}

The main result of this section is an embedding with prioritized distortion {\em and dimension}. This means that a high ranking point will have low distortion (with any other point), and additionally, its image will consist of few nonzero coordinates, followed by zeros in the rest.

\begin{theorem}
\label{thm: Priority distortionn and space}
For any $p\in[1,\infty]$, any fixed $\epsilon>0$ and any metric space $(X,d)$ on $n$ points,
there exists an embedding of $X$ into $\ell_{p}^{O(\log^2n)}$ with priority distortion
$O\left(\log^{4+\epsilon}j\right)$, and prioritized dimension $O(\log^4j)$.
\end{theorem}


\paragraph{Proof overview.} The basic framework of this embedding appears at a first glance to be similar to \sectionref{sec:distortion-lp}, which is applying a variation of Bourgain's embedding, while sampling only from certain subsets $S_i$ of the points. However, the crux here is that we need to ensure that high priority points will be mapped to the zero vector in the embeddings that "handle" the lower ranked points.

Recall that the coordinates of the embedding are given by distances to sets. The idea is the following: while creating the embedding for the points in $S_i$, we insert all the points with higher ranking (those in $S_0\cup\dots\cup S_{i-1}$) into every one of the randomly sampled sets. This will certify that the high ranked points are mapped to zero in every one of these coordinates. However, the analysis of the distortion no longer holds, as the sets are not randomly chosen. Fix some point $u\in S_i$ and $v\in X$. The crucial observation is that if none of the higher ranked points lie in certain neighborhoods around $u$ and $v$ (the size of these neighborhoods depends on $d(u,v)$), then we can still use the randomness of the selected sets to obtain some bound (albeit not as good as the standard embedding achieves). While if there exists a high ranked point nearby, say $z\in S_{i'}$ for some $i'<i$, then we argue that $u,v$ should already have sufficient contribution from the embedding designed for $S_{i'}$. The formal derivation of this idea is captured in \lemmaref{lem: Partial-Bourgain-Embedding}.

The calculation shows that the distortion guarantee for $u,v$ deteriorates by a logarithmic factor for each $i$, that is, it is the product of the distortion bound for points in $S_{i-1}$ multiplied by $O(\log|S_i|)$. This implies that the optimal size of $S_i$ is {\em triple exponential} in $i$, which yields the best balance between the price paid due to the size of $S_i$ and the product of the logarithms of $|S_0|,\dots,|S_{i-1}|$.

\begin{lemma}
\label{lem: Partial-Bourgain-Embedding}
Let $p\in[1,\infty]$ and $D\ge 1$. Given a metric space $\left(X,d\right)$, two disjoint subsets $A,K\subseteq X$
where $\left|K\right|=k\ge 2$, and a non-expansive embedding $g:X\rightarrow\ell_{p}$ with contraction at most
$D$ for all pairs in $A\times X$, then there is a non-expansive embedding
$f:X\rightarrow\ell_{p}^{O(\log^2 k)}$ such that the following properties hold:
\begin{enumerate}
\item \label{enu:- f(x)=00003D0 for x in A} For all $x\in A$, $f\left(x\right)=\vec{0}$~.
\item \label{enu: contraction bound} For all $(x,y)\in K\times X$, $\|f(x)-f(y)\|_p\ge \frac{d(x,y)}{1000D\cdot \log k}$, or
    $\|g(x)-g(y)\|_p\ge \frac{d(x,y)}{2D}$~.
\end{enumerate}
\end{lemma}


We postpone the proof of \lemmaref{lem: Partial-Bourgain-Embedding} to \sectionref{sec:lemma10}, and prove \theoremref{thm: Priority distortionn and space} using the lemma.
\begin{proof}[Proof of Theorem \ref{thm: Priority distortionn and space}]
Let $I=\lceil \log\log\log n\rceil$. Let $S_0=\left\{ x_{1},x_{2},x_3,x_4\right\} $, and for $1\le i\le I$ let $S_{i}=\left\{x_j~:~ 2^{2^{2^{i-1}}} < j\le 2^{2^{2^{i}}}\right\}$.
Also define $S_{< i}=\bigcup_{0\le k< i}S_k$.

The desired embedding $F:X\to\ell_p$ will be created by iteratively applying \lemmaref{lem: Partial-Bourgain-Embedding}, each time using its output function $f$ as part of the input for the next iteration. Formally, for each $0\le i\le I$ apply \lemmaref{lem: Partial-Bourgain-Embedding} with parameters $A=S_{<i}$, $K=S_i$, $g=F^{(i-1)}$ and $D=2^{2^i+5i^2}$, to obtain a map $f_i:X\to\ell_p$. The map $F^{(i)}:X\to\ell_p$ is defined as follows: $F^{(-1)}\equiv 0$, and $F^{(i)}=\bigoplus_{k=0}^i\alpha_k\cdot f_k$, where $(\alpha_k)$ is a sequence that ensures $F^{(i)}$ is non-expansive for all $i$. For concreteness, take $\alpha_k=\left(\frac{6}{\pi^2(k+1)^2}\right)^{1/p}$. The final embedding is defined by $F=F^{(I)}$.

Fix any pair $x,y\in X$. As $f_i$ is non-expansive by \lemmaref{lem: Partial-Bourgain-Embedding}, we obtain that $F$ is non-expansive as well:
\[
\|F(x)-F(y)\|_p^p= \sum_{i=0}^I\alpha_i^p\cdot \|f_i(x)-f_i(y)\|_p^p\le\sum_{i=0}^\infty\frac{6}{\pi^2(i+1)^2}\cdot d(x,y)^p=d(x,y)^p~.
\]
Next, we must show that for each $0\le i\le I$, the embedding $F^{(i-1)}$ has contraction at most $2^{2^i+5i^2}$ for pairs in $S_{<i}\times X$, to comply with the requirement of \lemmaref{lem: Partial-Bourgain-Embedding}. We prove this by induction on $i$, the base case for $i=0$ holds trivially as $F^{(-1)}$ has no requirement on its contraction (since $S_{<0}=\emptyset$). Assume (for $i$) that $F^{(i-1)}$ has contraction at most $2^{2^i+5i^2}$ on pairs in $S_{<i}\times X$. For $i+1$, let $x\in S_{<i+1}$ and $y\in X$. Recall that $F^{(i)}$ is generated by applying \lemmaref{lem: Partial-Bourgain-Embedding} with $A=S_{<i}$, $K=S_i$, $g=F^{(i-1)}$, and $D=2^{2^i+5i^2}$. Then the lemma returns $f_i$, and finally $F^{(i)}=g\oplus(\alpha_i\cdot f_i)$.

We may assume that $x\in S_i$, otherwise $g=F^{(i-1)}$ has the required contraction on $x,y$ by the induction hypothesis. Applying condition (\ref{enu: contraction bound}) of the lemma: if it is the case that $\|g(x)-g(y)\|_p\ge d(x,y)/(2D)$, then clearly $2D<  2^{2^{i+1}+5(i+1)^2}$. The other case is that $\|f_i(x)-f_i(y)\|_p\ge \frac{d(x,y)}{1000D\cdot \log |S_i|}$. Since $\log|S_i|\le 2^{2^i}$ and $1/\alpha_i\le 2(i+1)^2$, the contraction of $F^{(i)}$ is at most the contraction of $\alpha_i\cdot f_i$, which is bounded by
\[
\frac{1000D\cdot \log|S_i|}{\alpha_i}\le 1000\cdot 2^{2^i+5i^2}\cdot 2^{2^i}\cdot 2(i+1)^2<2^{2\cdot 2^i+ 5i^2+2\log(i+1)+11}<2^{2^{i+1}+5(i+1)^2}~.
\]

Observe that if $x=x_j\in S_i$ for some $j>1$, then $2^{2^{i-1}}<\log j$, and thus the distortion of $F$ for any pair containing $x$ is at most $2^{2^{i+1}+5(i+1)^2}=O(\log^4j)\cdot 2^{O((2+\log\log\log j)^2)}=O(\log^{4+\epsilon}j)$. Additionally, note that as the distortion of $F^{(I-1)}$ is at most $D=2^{2^I+5I^2}$, the same argument suggests that the maximal distortion of $F=F^{(I)}$ for any pair is at most
\[
\frac{1000D\cdot\log n}{\alpha_I}\le 1000\cdot 2^{2^I+5I^2}\cdot\log n\cdot 2(I+1)^2=  O(\log^{3+\epsilon}n)~.
\]

Finally, let us bound the number of nonzero coordinates of the points. Recall that $f_i$ maps $X$ into $O(\log^2|S_i|)\le O(2^{2^{i+1}})$ dimensions. Fix some $x=x_j$ for $j>1$, and let $i$ be such that $x_j\in S_i$. Note that $2^{2^{i-1}}<\log j$, so that $2^{2^{i+1}}<\log^4j$. By \lemmaref{lem: Partial-Bourgain-Embedding}, for every $i'>i$, $f_{i'}(x_j)=\vec{0}$, and the number of coordinates used by $F^{(i)}$ is at most
\[
\sum_{k=0}^i O(2^{2^{k+1}}) = O(2^{2^{i+1}}) = O(\log^4j)~.
\]

Since the dimension of $f_I$ is at most $O(\log^2n)$, we get that the total number of coordinates used by $F$  is only
\[
\sum_{k=0}^{I-1} O(2^{2^{k+1}}) +O(\log^2n) \le O(2^{2^{1+\log\log\log n}}) + O(\log^2n)=O(\log^2n)~.
\]

\end{proof}

\subsubsection{Proof of \lemmaref{lem: Partial-Bourgain-Embedding}}\label{sec:lemma10}

The basic approach to the proof is similar to \lemmaref{lemma:lp-term-expand}, which is sampling subsets of $K$, according to various densities. The main difference is that we insert all the points of $A$ into each sampled set, to ensure $f(x)=\vec{0}$ for all $x\in A$. The standard analysis of Bourgain for a pair $x,y$, considers certain neighborhoods defined according to the density of points around $x,y$. We show that the analysis still works as long as no point of $A$ is present in those neighborhoods. Thus we can obtain a contribution which is proportional to the distance of $x,y$ to $A$ (or to $d(x,y)$ if that distance is large). This motivates the following definition and lemma.

\begin{definition}
The $\gamma$-distance between $x$ and $y$ with respect to $A$ is defined to be
\[
\gamma_{A}\left(x,y\right)=\min\left\{ \frac{d(x,y)}{2},d(x,A),d(y,A)\right\}~.
\]
\end{definition}

\begin{lemma}
\label{lemma:main Partial Bourgain Embedding} Let $c=24$. There exists a non-expansive embedding $\varphi:X\to\ell_p^{O(\log^2k)}$, such that for all $z\in A$, $\varphi(z)=\vec{0}$, and for all $x,y\in K$,
\[
\|\varphi(x)-\varphi(y)\|_p\ge \frac{\gamma_A(x,y)}{c\log k}~.
\]
\end{lemma}
We defer the proof of \lemmaref{lemma:main Partial Bourgain Embedding}, and proceed first with the proof of \lemmaref{lem: Partial-Bourgain-Embedding}.
Define $h:X\to\R$ for $x\in X$ as $h(x)=d(x,A\cup K)$. Our embedding $f$ is
\[
f=\frac{\varphi\oplus h}{2^{1/p}}~.
\]
Since both $\varphi$ and $h$ are non-expansive and vanish on $A$, clearly $f$ is non-expansive as well, and $f(z)=\vec{0}$ for any $z\in A$. It remains to show property (\ref{enu: contraction bound}) of the lemma.
Fix any $x\in K$ and $y\in X$, and consider the following three cases:
\paragraph{Case 1:} $d\left(\left\{ x,y\right\} ,A\right)\le\frac{d(x,y)}{4D}$.

In this case we shall use the guarantees of the map $g$. Assume w.l.o.g that $z\in A$ is such that $d(y,z)\le \frac{d(x,y)}{4D}$. Then by the triangle inequality
\begin{equation}\label{eq:rree}
d(x,z)\ge d(x,y)-d(y,z)\ge d(x,y)-\frac{d(x,y)}{4D}\ge\frac{3d(x,y)}{4}~.
\end{equation}
Now, using that $g$ is non-expansive, and has contraction at most $D$ for any pair in $A\times X$, we obtain that
\begin{eqnarray*}
\|g(x)-g(y)\|_p&\ge& \|g(x)-g(z)\|_p - \|g(z)-g(y)\|_p\\
&\ge&\frac{d(x,z)}{D} - d(z,y)\\
&\stackrel{\eqref{eq:rree}}{\ge}&\frac{3d(x,y)}{4D} - \frac{d(x,y)}{4D}\\
&=&\frac{d(x,y)}{2D}~,
\end{eqnarray*}
which satisfies property (\ref{enu: contraction bound}).

\paragraph{Case 2:} $d\left(\left\{ x,y\right\} ,A\right)>\frac{d(x,y)}{4D}$ and $d(y,K)\ge\frac{d(x,y)}{20cD\cdot\log k}$ (where $c=24$ is the constant of \lemmaref{lemma:main Partial Bourgain Embedding}).

Here we shall use the map $h$ for the contribution. Since $d(y,A)\ge d(x,y)/(4D)$, we have that $h(y) = d(y,A\cup K)\ge \frac{d(x,y)}{20cD\cdot\log k}$
and of course $h(x)=0$, so that
\[
\|f(x)-f(y)\|_p\ge \frac{|h(x)-h(y)|}{2}\ge\frac{d(x,y)}{40cD\cdot\log k}~,
\]
as required.

\paragraph{Case 3:} $d\left(\left\{ x,y\right\} ,A\right)>\frac{d(x,y)}{4D}$ and $d(y,K)<\frac{d(x,y)}{20cD\cdot\log k}$.

In this case, the function $\varphi$ will yield the required contribution, by employing a similar strategy to \lemmaref{lemma:lp-term-expand}. Let $k_y\in K$ be such that $d(y,k_y)=d(y,K)$. Note that $d(k_y,A)\ge d(y,A)-d(y,k_y)\ge\frac{d(x,y)}{4D}-\frac{d(x,y)}{20cD\cdot\log k}\ge \frac{d(x,y)}{5D}$, and it follows that
\begin{equation}\label{eq:trtrt}
\gamma_A(x,k_y)\ge\frac{d(x,y)}{5D}~.
\end{equation}
By \lemmaref{lemma:main Partial Bourgain Embedding}, since $f$ is non-expansive, and using another application of the triangle inequality, we conclude that
\begin{eqnarray*}
\|f(x)-f(y)\|_p&\ge&\|f(x)-f(k_y)\|_p-\|f(y)-f(k_y)\|_p\\
&\ge&\frac{\|\varphi(x)-\varphi(k_y)\|_p}{2}-d(y,k_y)\\
&\ge&\frac{\gamma_A(x,k_y)}{2c\log k}-\frac{d(x,y)}{20c D\cdot\log k}\\
&\stackrel{\eqref{eq:trtrt}}{\ge}&\frac{d(x,y)}{10c D\cdot\log k}-\frac{d(x,y)}{20c D\cdot\log k}\\
&=&\frac{d(x,y)}{20c D\cdot\log k}~.
\end{eqnarray*}
This concludes the proof of \lemmaref{lem: Partial-Bourgain-Embedding}. It remains to validate \lemmaref{lemma:main Partial Bourgain Embedding}, which is similar in spirit to the methods of \cite{bourgain,llr}, we give full details for completeness.

\begin{proof}[Proof of Lemma \ref{lemma:main Partial Bourgain Embedding}]
Let $I=\lceil \log k\rceil$ and $J=C\cdot\log k$ for a constant $C$ that will be determined later.
For each $i\in[I]$ and $j\in[J]$ sample a set $Q'_{ij}$ by including each $x\in K$ independently with probability $2^{-i}$, and let $Q_{ij}=Q'_{ij}\cup A$. Define maps $\varphi_{ij}:X\to\R$ by letting for each $u\in X$, $\varphi_{ij}(u)=d(u,Q_{ij})$, and $\varphi:X\to\ell_p^{I\cdot J}$ by
\[
\varphi(u)=\frac{1}{(I\cdot J)^{1/p}}\bigoplus_{i\in[I]}\bigoplus_{j\in[J]}\varphi_{ij}(u)~.
\]
Since each $\varphi_{ij}$ is non-expansive, $\varphi$ is non-expansive as well, and in what follows we bound its contraction.

Define for $u\in K$ and $r\ge 0$ the ball restricted to $K$, $B_K(u,r)=B(u,r)\cap K$, and recall that by $B^{\circ}$ we mean the open ball.
Fix a pair $u,v\in K$, and for each $0\le i\le I$, let $r'_i$ be the minimal such that both $|B_K(u,r)|\ge 2^i$ and $|B_K(v,r)|\ge 2^i$. Define $r_i=\min\{r'_i,\gamma_A(u,v)\}$ and let $\Delta_i=r_i-r_{i-1}$. Observe that $r_0=0$ and $r_I=\gamma_A(u,v)$, so that
\begin{equation}\label{eq:frf}
\sum_{i\in[I]}\Delta_i = \gamma_A(u,v)~.
\end{equation}
We first claim that for each $i\in[I]$ and $j\in[J]$,
\begin{equation}\label{eq:poi}
\Pr[|\varphi_{ij}(u)-\varphi_{ij}(v)|\ge \Delta_i]\ge 1/12~.
\end{equation}
If $\Delta_i=0$ then there is nothing to prove. Assume then that $r_{i-1}<r_i$, and note that either $|B_K^\circ(u,r_i)|\le 2^i$ or $|B_K^\circ(v,r_i)|\le 2^i$ (otherwise it contradicts the minimality of $r_i$). W.l.o.g we have that $|B_K^\circ(u,r_i)|\le 2^i$. Furthermore, note that the sets $B_K^\circ(u,r_i)$, $B_K(v,r_{i-1})$ and $A$ are pairwise disjoint. Let ${\cal E}$ be the event that $\{Q_{ij}\cap B_K^\circ(u,r_i)=\emptyset\}$ and ${\cal F}$ be the event that $\{Q_{ij}\cap B_K(v,r_{i-1})\neq\emptyset\}$. Observe that if both events hold then $d(u,Q_{ij})\ge r_i$ and $d(v,Q_{ij})\le r_{i-1}$, so that
\[
|\varphi_{ij}(u)-\varphi_{ij}(v)|\ge r_i-r_{i-1}=\Delta_i~.
\]
Since both balls are disjoint from $A$, we have that
\begin{eqnarray*}
\Pr[{\cal E}]  =  \prod_{x\in B_{K}^\circ\left(u,r_{i}\right)}\Pr\left[x\notin Q'_{ij}\right]
=\left(1-2^{-i}\right)^{\left|B_{K}^\circ\left(u,r_{i}\right)\right|}
 \ge  \left(1-2^{-i}\right)^{2^{i}}\ge\frac{1}{4}\ .
\end{eqnarray*}
And similarly,
\begin{eqnarray*}
\Pr[{\cal F}] =  1-\prod_{x\in B_{K}\left(v,r_{i-1}\right)}\Pr\left[x\notin Q'_{ij}\right]
=1-\left(1-2^{-i}\right)^{\left|B_{K}\left(v,r_{i-1}\right)\right|}
 \ge 1-\left(1-2^{-i}\right)^{2^{i-1}}
\ge 1-e^{-\frac{1}{2}}\ge\frac{1}{3}~ .
\end{eqnarray*}
Since the events ${\cal E}$ and ${\cal F}$ are independent, this concludes the proof of \eqref{eq:poi}.
Let $X_{ij}$ be an indicator random variable for the event that $|\varphi_{ij}(u)-\varphi_{ij}(v)|\ge \Delta_i$, and $X_i=\sum_{j=1}^JX_{ij}$. Using the independence for different values of $j$, and that $\E[X_i]\ge J/12$, a Chernoff bound yields that for any $i$
\[
\Pr[X_i< J/24]\le e^{-J/100}\le 1/k^3~,
\]
when $C$ is sufficiently large.
Note that if indeed $X_i\ge J/24$ for all $1\le i\le I$ then
\begin{eqnarray*}
\|\varphi(u)-\varphi(v)\|_{p}^{p} & = & \frac{1}{I\cdot J}\sum_{i=1}^I\sum_{j=1}^J|\varphi_{ij}(u)-\varphi_{ij}(v)|^{p}\\
& \ge & \frac{1}{24I}\sum_{i=1}^I\Delta_{i}^{p}\\
& \ge & \frac{I^{1-p}}{24I}\Big(\sum_{i=1}^I\Delta_{i}\Big)^{p}\\
&\stackrel{\eqref{eq:frf}}{\ge}& \frac{\gamma_A(u,v)^p}{24I^p}~,
\end{eqnarray*}
where the second inequality uses H\"{o}lder's inequality.
Applying a union bound over the ${k\choose 2}$ possible pairs in ${K\choose 2}$, and the $I=\lceil\log k\rceil$ possible values of $i$, there is at least a constant probability that for every pair $\|\varphi(u)-\varphi(v)\|_{p}\ge \frac{\gamma_A(u,v)}{24^{1/p}\cdot \log k}$.

\end{proof}
\bibliographystyle{alpha}
\bibliography{art}

\end{document}